%% file: arxiv_20250103.tex
\documentclass[12pt]{article}
\pdfoutput=1
\usepackage[margin=1in]{geometry}

\usepackage[
  bookmarks=true,
  bookmarksnumbered=true,
  bookmarksopen=true,
  pdfborder={0 0 0},
  breaklinks=true,
  colorlinks=true,
  linkcolor=black,
  citecolor=black,
  filecolor=black,
  urlcolor=black,
]{hyperref}

\usepackage[utf8]{inputenc}
\usepackage{titlesec}
\usepackage{comment}
\usepackage{amsmath}
\usepackage{amsfonts}
\usepackage{amsthm}
\usepackage{subcaption}
\usepackage[round]{natbib}

\newtheorem{thm}{Theorem}[section]

\newtheorem{lemma}[thm]{Lemma}

\newtheorem{property}[thm]{Property}
\newtheorem{remark}[thm]{Remark}

\theoremstyle{definition}
\newtheorem{definition}[thm]{Definition}

\newcommand\numberthis{\addtocounter{equation}{1}\tag{\theequation}}

\usepackage{algorithm}
\usepackage{algpseudocode}
\usepackage{parskip}
\usepackage{setspace}
\algnewcommand{\LineComment}[1]{\State \(\triangleright\) #1}

\usepackage{amsmath,amsthm}
\usepackage{mathtools}
\usepackage{amssymb}
\usepackage{amsfonts}
\usepackage{graphicx}
\usepackage{tikz}

\usepackage{parskip}

\newtheorem{problem}{Problem}
\newtheorem{observation}{Observation}

\theoremstyle{definition}

\DeclareMathOperator{\E}{\mathbb{E}}
\DeclareMathOperator{\alg}{\mathrm{ALG}}

\DeclarePairedDelimiterX{\frob}[2]{\langle}{\rangle_F}{#1, #2}
\newcommand\norm[1]{\lVert#1\rVert}

\title{Improved Regret Bounds for \\ Online Fair Division with Bandit Learning\thanks{Schiffer was supported by an NSF Graduate Research Fellowship. Zhang was supported by an NSF Graduate Research Fellowship.}}

\author{
	Benjamin Schiffer\thanks{Department of Statistics, Harvard University | \emph{E-mail}: \href{mailto:bschiffer1@g.harvard.edu}{bschiffer1@g.harvard.edu}.}
	\and
	Shirley Zhang\thanks{Paulson School of Engineering and Applied Sciences, Harvard University | \emph{E-mail}: \href{mailto:szhang2@g.harvard.edu}{szhang2@g.harvard.edu}.}
}
\begin{document}

\begin{titlepage}
\maketitle

\setcounter{page}{0}
\thispagestyle{empty}

\begin{abstract}
We study online fair division when there are a finite number of item types and the player values for the items are drawn randomly from distributions with unknown means. In this setting, a sequence of indivisible items arrives according to a random online process, and each item must be allocated to a single player. The goal is to maximize expected social welfare while maintaining that the allocation satisfies proportionality in expectation. When player values are normalized, we show that it is possible to with high probability guarantee proportionality constraint satisfaction and achieve $\tilde{O}(\sqrt{T})$ regret. To achieve this result, we present an upper confidence bound (UCB) algorithm that uses two rounds of linear optimization. This algorithm highlights fundamental aspects of proportionality constraints that allow for a UCB algorithm despite the presence of many (potentially tight) constraints. This result improves upon the previous best regret rate of $\tilde{O}(T^{2/3})$.
\end{abstract}

\end{titlepage}

\input{body}

\newpage
\bibliographystyle{abbrvnat}
\bibliography{abb,ultimate, bib}

\newpage
\appendix

\include{appendix}

\end{document}

%% file: body.tex
\section{Introduction}

The fair division of indivisible goods is a classic problem in computational social choice. In this problem, a set of goods must be fairly allocated among $n$ players, where each player may have a different value for each good. For example, consider a central food bank which is in charge of distributing food to multiple food pantries across the region. Each food pantry may have differing preferences over various types of food depending on the populations they serve. The goal is then to divide the goods in a way that is fair relative to each player's valuations.  

Online fair division adds another degree of difficulty. Instead of all items being known upfront, in online fair division goods arrive one at a time and must be irrevocably allocated at the time of arrival. The first goods that arrive must be allocated without knowing what future goods will be, and cannot be reallocated after more goods arrive. In the food bank example, this setting is especially relevant for perishable goods, which the food bank must quickly allocate after arrival. For example, \citep{mertzanidis2024automating} describes a partnership with an Indiana program which redistributes rejected food by redirecting truck drivers from landfills to food banks. Truck drivers arrive on the app in an online manner, necessitating an online fair division algorithm to match drivers to food banks. The food available depends on what was rejected on a given day, which may be unpredictable. 

A common fairness requirement is proportionality, which insists that each player receive at least $\frac{1}{n}$ of their total value for all goods. In settings where the allocation of an item may be randomized, it is natural to instead consider proportionality in expectation, where the expectation is taken over both the random player values and the random item types. For a given instance, there may be many proportional allocations, in which case we can differentiate further by also considering efficiency. Past works have incorporated efficiency by, for example, requiring Pareto optimality in addition to proportionality \citep{BKPP+24} or maximizing utilitarian social welfare subject to proportionality \citep{procaccia2024honor}.

We study the online fair division problem in the setting where player values are \textit{unknown} at the time of item arrival, and the value of a player for an item is only revealed if the item is allocated to that player. Specifically, we consider the setting where there are $m$ item types, and each player's value for an item of each type is drawn from an (possibly different) unknown distribution. We would like to distribute items fairly in expectation to players despite not knowing their true mean values for item types. In our food bank example, this setting is most relevant when a food bank does not yet understand the needs of its food pantries, but can easily collect information from each food pantry (e.g. in the form of surveys) regarding how well items are received by each food pantry's visitors. \citep{yamada2024learning} gives further motivation for this setting in the form of allocating difficult tasks to users with different strengths and distributing humanitarian aid, both of which are online fair division problems in which values are only unveiled after requesting feedback. In such settings, it is necessary to learn each player's expected value for each item. 

We consider the problem statement proposed in \citep{procaccia2024honor}, which is as follows. Subject to the fairness constraints of proportionality in expectation, we strive for efficiency as measured by the \textit{expected utilitarian social welfare} of our solution. Specifically, we aim to minimize the regret incurred when comparing to an optimal solution which must adhere to the fairness constraint but knows the true means. As the algorithm does not know the players' true means, it will be impossible to achieve non-trivial regret and guarantee proportionality in expectation at every time step. Therefore, we will instead require that the algorithm with high probability satisfies the proportionality in expectation constraints at every time step. 

In summary, we consider online fair division with unknown means as a reinforcement learning problem subject to fairness constraints. In this paper, we study how to balance exploration and exploitation while also maintaining proportionality.

 \subsection{Our Contributions}
In this work, we study the problem of online fair division with unknown means subject to proportionality constraints. Our main result is that, when player values are normalized, an algorithm can achieve $\tilde{O}(\sqrt{T})$ regret while satisfying proportionality in expectation constraints at every time step with high probability (Theorem \ref{thm:sqrtT_regret}). This is an improvement on the $\tilde{\Omega}(T^{2/3})$ regret of the explore-then-commit algorithm in \citep{procaccia2024honor}. The algorithm that achieves our result (Algorithm \ref{algo:UCB}) uses a variant of upper confidence bound (UCB) logic with two rounds of linear program optimization. The first round of optimization guarantees that with high probability the constraints are satisfied, but does not provide sufficient exploration for UCB. Therefore, Algorithm \ref{algo:UCB} performs a second round of optimization that exploits the underlying structure of the fairness constraints to sufficiently explore without losing significant social welfare.

We complement our positive results for proportionality with an impossibility result for envy-freeness, another commonly studied fairness notion. Specifically, we show that when values are normalized, the best regret rate for envy-freeness is $\tilde{O}(T^{2/3})$, which matches the lower bound for envy-freeness when values are not normalized. This highlights a fundamental difference in the difficulty of maintaining envy-freeness versus proportionality when learning unknown values.

\subsection{Related Work}

Our problem is most closely related to that in \citep{procaccia2024honor}, which introduces the  problem setting we study. \citep{procaccia2024honor} studies both envy-freeness in expectation and proportionality in expectation constraints, and provide explore-then-commit algorithms which achieve $\tilde{O}(T^{2/3})$ regret while maintaining these fairness constraints with high probability. \citep{procaccia2024honor} also proves that no algorithm can have lower regret while maintaining these fairness constraints. In contrast, we show that when players' values are normalized, there do exist algorithms which achieve $\tilde{O}(\sqrt{T})$ regret while maintaining proportionality in expectation at each time step with high probability. The algorithm we present is an upper confidence bound  algorithm rather than an explore-then-commit algorithm as used in \citep{procaccia2024honor}. While our algorithm relies on fundamental properties of fairness constraints similar to those in \citep{procaccia2024honor}, the properties needed for $\tilde{O}(\sqrt{T})$ regret are stronger and are \textit{not} satisfied by envy-freeness constraints. 

We briefly mention two other related works in online fair division. \citep{yamada2024learning} studies a similar setting in which a player's value for an item type is unknown and is only observed (with noise) when an item is allocated to that player. Rather than encoding fairness as a constraint, however, \citep{yamada2024learning} instead maximizes Nash social welfare in the objective function by leveraging algorithms which use dual averaging. \citep{BKPP+24} studies a somewhat different online fair division setting in which items arrive in an adversarial manner, but the values of all agents for each item are known when the item arrives (and, crucially, before item allocation).\citep{BKPP+24} also considers efficiency, but in the form of guaranteeing Pareto optimality rather than maximizing utilitarian social welfare. 

There are many notions of fairness that have been studied in the multi-armed bandits literature. One such notion is that similar individuals are treated similarly ~\citep{chen2021fairer, liu2017calibrated}. Another related notion of fairness is that `worse' arms are never pulled with higher probability than 'better' arms ~\citep{joseph2016fairness, joseph2016fair}. These two notions of fairness are incompatible with proportionality, as proportionality may require that worse arms are pulled with higher probability. A third notion of fairness is the requirement that every arm is pulled a minimum proportion of the time~\citep{chen2020fair,claure2020multi,li2019combinatorial,patil2021achieving}. Once again, this notion of fairness is not compatible with proportionality, as there exist proportional allocations where every item type is not allocated to every individual. We also focus solely on the non-contextual setting, however many works also study fairness when there is context \citep{grazzi2022group, schumann2019group,wang2021fairness, wu2023best, wei2024fair}.

Because our fairness constraints and objective function are linear, our problem formulation is also related to the problem of multi-armed bandits under general linear constraints. One area of work studies linear bandits under linear safety constraints~\citep{amani2019linear, carlsson2024pure, moradipari2020linear}. \citep{amani2019linear} shows that if there is a single linear constraint and the optimal solution has positive slack with respect to this constraint, then $\tilde{O}(\sqrt{T})$ regret is possible. When the slack of the optimal solution is $0$, \citep{amani2019linear} presents an algorithm that has regret of $\tilde{O}(T^{2/3})$. Our setting differs from \citep{amani2019linear} in that we have many linear constraints (one for each player), and the optimal solution frequently has zero slack for some of the constraints. The existence of the $\tilde{O}(\sqrt{T})$ regret algorithm in our setting despite these added difficulties fundamentally relies on the structure of our constraints. Other works studying linear bandits have constraints that differ from our setting either because they are not applied at every time step or because they require slack in the constraints~\citep{liu2021efficient, pacchiano2021stochastic}. 

Another related area is bandits with knapsacks, which also studies bandits problems with constraints~\citep{liu2022combinatorial,badanidiyuru2018bandits}. However, the knapsack constraints depend on resource consumption vectors rather than the unknown mean values, and therefore these constraints are significantly different than proportionality constraints. 

\section{Model}

In our setting, there are $N = [n]$ players and $M = [m]$ item types, with $T$ items arriving over time. For any item of type $k$, the value of player $i$ for that item is drawn from a sub-Gaussian distribution with mean $\mu^*_{ik}$. For each round $t$ from $0$ to $T-1$, an item $j_t$ of type $k_t$ is drawn uniformly at random. As shown in \citep{procaccia2024honor}, the assumption of uniformly random item types is WLOG, and all of our results hold when item types are drawn from any arbitrary distribution. After observing $k_t$, an algorithm allocates the item $j_t$ (potentially randomly) to a player $i_t$, and then observes $i_t$'s value $v_t$ for $j_t$. In order to decide how to allocate $j_t$, an algorithm may consult the history $H_t = \{k_{t'}, i_{t'}, v_{t'}\}_{t' < t}$. Note that the algorithm never observes the values of a player $i$ for item $j$ if $j$ is not allocated to $i$.  

An algorithm allocates items to players via fractional allocations $X \in \mathbb{R}^{n \times m}$, where a fractional allocation is said to be \textit{valid} if $\sum_i X_{ik} = 1$ for all $k \in [m]$. Intuitively, the $k$th column of a fractional allocation represents how the algorithm will randomly allocate the item if the item has type $k$. One valid fractional allocation is the \textit{uniform at random} (UAR) allocation, where every element is equal to $\frac{1}{n}$. At time $t$, before observing $k_t$, an algorithm considers the history $H_t$ and outputs a valid fractional allocation $X^t = \alg(H_t)$. After observing $k_t$, the algorithm will then allocate $j_t$ based on the probabilities in $((X^t)^{\top})_{k_t}$, i.e. the $k_t$th column of $X^t$. The expected value of player $i$ for a fractional allocation $X$ can be written as $\frac{1}{m} X_i \cdot  \mu^*_i$, and the sum over all players' expected values for $X$ is then 
\[
\frac{1}{m} \sum_{i \in [n]} X_i \cdot  \mu^*_i  = \frac{1}{m} \frob{X}{\mu^*}.
\]
Therefore, we can write the total expected social welfare of an algorithm which allocates according to $X^t$ at time $t$ as
\[
\E[\text{social welfare of $\alg$}] = \frac{1}{m} \sum_{t = 0}^{T - 1} \frob{X^t}{\mu^*}.
\]

In this paper, we make two additional assumptions on the unknown mean matrix $\mu^*$. The first assumption is that there exist known $a,b \in \mathbb{R}$ such that  $0 < a \le \mu_{ik}^* \le b < \infty$ for all $i,k$. As shown in \citep{procaccia2024honor}, this is a necessary assumption in order to achieve $o(T)$ regret. The second assumption made in this paper is that the values for each player are normalized. In other words, we assume that for all players $i$, $\sum_{k=1}^m \mu_{ik}^* = 1$. Informally, normalizing values ensures that each player has equal say in the total social welfare. Normalized values is a standard assumption in fair division literature (see, e.g., \citep{gkatzelis2021fair,bogomolnaia2022fair, yamada2024learning}) and further justification can be found in \citep{aziz2020justifications}. Note that assuming normalized values does not affect the proportionality constraints, which are invariant to scaling.

\subsection{Regret and Problem Formulation}
In this section, we give our formal definitions for fairness and regret.

The main fairness notion we study is \textit{proportionality in expectation}. Proportionality in expectation requires that, for every $t$, every player's expected value for the allocation $X^t$ is at least as much as that player's expected value for the uniform at random allocation.

Formally, player $i$'s expected value for fractional allocation $X^t$ is equal to $\frac{1}{m} X^t_i \cdot \mu^*_i$, where the $1/m$ comes from every item having probability $1/m$ of being each item type. Player $i$'s expected value for the UAR allocation is $\frac{1}{nm}\norm{\mu^*_i}_{1} = \frac{1}{nm}$ due to the normalized values assumption. Therefore, a fractional allocation is proportional in expectation if and only if  $\frac{1}{m} X^t_i \cdot \mu^*_i \ge \frac{1}{nm}$, which is equivalent to $X^t_i \cdot \mu^*_i \ge \frac{1}{n}$. We define proportionality in expectation for an algorithm $\alg$ in Definition \ref{def:proportionality}.
\begin{definition}\label{def:proportionality}
    An algorithm $\alg$ that uses fractional allocation $X^t$ at time $t$ satisfies \emph{proportionality in expectation} for $\mu^*$ if
    \[
    X^t_i \cdot \mu^*_i \ge \frac{1}{n} \quad \forall t < T, i \in [n].
    \]
\end{definition}

In this paragraph we will briefly summarize from \citep{procaccia2024honor} the justification for studying proportionality in expectation rather than realized proportionality. First, note that satisfying proportionality in expectation does not guarantee that every player prefers their final set of allocated items at time $T$ to a $1/n$ proportion of all $T$ items (which we refer to as realized proportionality). In fact, no algorithm can with high probability guarantee that every player prefers their final set of allocated items at time $T$ to a $1/n$ proportion of all $T$ items. Define the \textit{dis-proportionality} of the final allocation as the maximum across all players of how much each player prefers a $1/n$ proportion of all $T$ items to their allocated items at time $T$. An algorithm that satisfies proportionality in expectation has the asymptotically optimal rate of dis-proportionality among all possible algorithms.  See \citep{procaccia2024honor} for more discussion about the optimality of proportionality in expectation, including formal statements and proofs.

For any value matrix $\mu$, let $Y^{\mu}$ be the expected social welfare maximizing fractional allocation that satisfies proportionality in expectation for $\mu$. Formally, we define 
\begin{align*}
    Y^\mu := \arg\max & \: \:  \frob{X}{\mu} \\
    \text{s.t. } & X^t_i \cdot \mu_i^* \ge \frac{1}{n} \quad \forall i \\
    & \sum_{i} X_{ik} = 1 \quad \forall k  \numberthis \label{lp:ymu}
\end{align*}
If the true mean values matrix $\mu^*$ is known, then the social welfare maximizing algorithm $\alg$ that satisfies proportionality in expectation is simply the algorithm that chooses $X^t = Y^{\mu^*}$ for all $t \in [0:T-1]$. Using this optimal algorithm as a baseline, we define the regret of an arbitrary algorithm $\alg$ as follows.
 \begin{definition}\label{def:regret}
   Define $Y^{\mu^*}$ as the solution to LP \eqref{lp:ymu} when $\mu = \mu^*$. Then the $T$-step regret for $\mu^*$ of an algorithm $\alg$ that uses allocation $X^t$ at time $t$ can be written as 
   \[
   \text{Regret of $\alg$ for $\mu^*$} = T \cdot \frob{Y^{\mu^*}}{\mu^*} - \sum_{t=0}^{T-1} \frob{X^t}{\mu^*}.
   \]
\end{definition}

The formal problem we address in this paper is as follows.
\begin{problem}\label{problem1}
    Design an algorithm $\alg$ such that the following result holds for any known $n,m, T,a,b$. Suppose that the true mean values satisfy $a \le \mu^*_{ik} \le b$ for all $i \in [n],k \in [m]$. Then with probability $1-1/T$, $\alg$ will both satisfy the proportionality in expectation constraints for $\mu^*$ and have regret for $\mu^*$ of $\tilde{O}(\sqrt{T})$.
\end{problem}

\subsection{Notation}
    Throughout this paper, we will use $O()$, $\tilde{O}(), \Omega(), \tilde{\Omega}()$ notation to represent the limiting behavior of functions with respect to $T$. For two matrices $A$ and $B$, we use $\frob{A}{B}$  to represent the Frobenius product of $A$ and $B$. We use $A_i$ to represent the $i$th row of matrix $A$ and $A_i \cdot B_i$ to represent the dot product between vectors $A_i$ and $B_i$. For matrices $\mu, \epsilon \in \mathbb{R}^{n \times m}$, define 
    \[
    B(\mu, \epsilon) = \{\mu' \in \mathbb{R}^{n \times m} : \mu_{ik} - \epsilon_{ik} \le \mu'_{ik} \le \mu_{ik} + \epsilon_{ik} \: \forall i,k  \}.
    \]

\section{Main Results}

\subsection{Algorithm Overview}
In this section, we present our main algorithm (Algorithm \ref{algo:UCB}) and main theorem (Theorem \ref{thm:sqrtT_regret}).

\begin{thm}\label{thm:sqrtT_regret}
  Suppose $n,m, T, a,b$ are known and that the true mean values satisfy $0 < a \le \mu^*_{ik} \le b$ for all $i \in [n]$, $k \in [m]$. With probability $1-1/T$, Algorithm \ref{algo:UCB} will both satisfy the proportionality in expectation constraints for $\mu^*$ and have regret for $\mu^*$ of $\tilde{O}(n^5m^3\sqrt{T})$.
\end{thm}

The initial exploration phase of Algorithm \ref{algo:UCB} uses the fact that the uniform at random allocation is guaranteed to satisfy proportionality in expectation constraints for any mean value matrix $\mu$. After the exploration phase, at each step we calculate an estimated mean value matrix ($\hat{\mu}^t$) and an uncertainty matrix ($\epsilon^t$). Algorithm \ref{algo:UCB} then performs two rounds of optimization. The first optimization of Algorithm \ref{algo:UCB} calculates an optimistic estimate of expected social welfare and guarantees that the solution $\hat{X}^t$ will satisfy the proportionality in expectation constraints for any $\mu \in B(\hat{\mu}^t, \epsilon^t)$. Therefore, if the true mean value matrix $\mu^*$ is in $B(\hat{\mu}^t, \epsilon^t)$, then $\hat{X}^t$ will satisfy the proportionality in expectation constraints for $\mu^*$. The algorithm unfortunately cannot directly use this $\hat{X}^t$ as the allocation in round $t$ because $\hat{X}^t$ does not necessarily provide sufficient exploration of all (item, player) pairs. See Section \ref{sec:discuss_Xhatt} below for more details. 

To avoid this issue, the algorithm includes a second round of optimization in LP \eqref{eq:lp_with_slack} that calculates an allocation $\hat{Z}^{ik}$ for each (item, player) pair $(i, k)$. $\hat{Z}^{ik}$ is guaranteed to sufficiently explore the $(i,k)$ pair and have social welfare that is not significantly less than the social welfare of $\hat{X}^t$. By using the fractional allocation $X^t$ that averages over all $\hat{Z}^{ik}$, the algorithm is able to guarantee that $X^t$ sufficiently explores every (player, item) pair. We note that due to the maximization in the second round of optimization, the runtime of Algorithm \ref{algo:UCB} is exponential in $n$ and $m$.  In the algorithm notation below, subscripts represent matrix indexing while superscripts represent matrix names, i.e. $Z^{ik}$ is a matrix, and $X^t_{ik}$ is the $(i,k)$ entry of $X^t$.

 \begin{algorithm}[htb!]
 \caption{UCB Online Fair Division}
 \label{algo:UCB}
 \begin{algorithmic}
\Require $n,m,T,a,b$
\For{$t \gets 0$ to $\log^2(T)\sqrt{T}-1$}
    \State Use $X^t = \mathrm{UAR}$.
\EndFor
\For{$t \gets \log^2(T)\sqrt{T}$ to $T$}
     \State $N^t_{ik} \gets  \sum_{\tau=0}^{t-1} 1_{k_\tau = k, i_\tau = i}$
     \State $\hat{\mu}^t_{ik} \gets  \frac{1}{N^t_{ik}} \sum_{\tau = 0}^{t-1} 1_{k_\tau = k, i_\tau = i} v_{\tau}$
     \State $\epsilon^t_{ik} = \sqrt{\log^2(6nmT)/(N^t_{ik})}$
     \State $(\mu_U^t)_{ik} = \hat{\mu}_{ik}^t + \epsilon_{ik}^t$
     \State $G^t = \left\{\mu \in B(\hat{\mu}^t, \epsilon^t) : \sqrt{T}\mu_{ik} \in \mathbb{Z} \quad \forall i,k\right\}$ 
     \State $\hat{X}^t \gets$ Solution  to the following LP:
             \begin{align*}
         \max_{X} \: &  \frob{X}{\mu^t_U} \\
        \text{s.t. } &  X_{i'} \cdot \mu_{i'} \ge \frac{1}{n}  \quad \forall i' \in [n], \: \forall \mu \in B(\hat{\mu}^t, \epsilon^t) \\
        & \sum_{i} X_{ik} = 1 \quad \forall k   \numberthis \label{eq:lp_etc_finallp2}
        \end{align*}  \label{line:hatX}
    \State $\forall i \in [n], \forall k \in [m]$, $\hat{Z}^{ik} \gets$ Solution to the following LP:
    \begin{align*}
         \max \:&  X_{ik} \\
        \text{s.t. } & X_{i'} \cdot \mu_{i'} \ge \frac{1}{n}  \quad \forall i' \in [n], \: \forall \mu \in B(\hat{\mu}^t, \epsilon^t) \\ 
        & \frob{X}{\mu^t_U} \ge \frob{\hat{X}^t}{\mu^t_U}  -   \frac{4bn}{a}\max_{\mu \in G^t} \frob{Y^{\mu}}{\epsilon^t} \\ 
        & \quad \quad \quad \quad \quad \quad- 2\frob{\hat{X}^t}{\epsilon_t}\\
        & \sum_{i'} X_{i'k'} = 1 \quad \forall k'  \numberthis \label{eq:lp_with_slack}
    \end{align*}
    \State Use $X^t = \frac{1}{nm} \left(\sum_{i,k} \hat{Z}^{ik} \right)$ 
\EndFor
 \end{algorithmic}
 \end{algorithm}

\subsection{Algorithm Intuition}\label{sec:discuss_Xhatt} 

In this section, we discuss the intuition behind the $\tilde{O}(\sqrt{T})$ regret guarantee of Algorithm \ref{algo:UCB}. First, we describe how Algorithm \ref{algo:UCB} relates to the standard multi-armed bandits upper confidence bound algorithm. We then describe the importance of the second round of optimization in Algorithm \ref{algo:UCB} and why the algorithm does not simply use the fractional allocation $\hat{X}^t$ at time $t$.

Consider the standard multi-armed bandits (MAB) setting, where there are $n$ arms that each have an unknown mean reward. At each time step, the algorithm chooses one arm, and the goal is to maximize the $T$-step reward. In this setting, a standard UCB algorithm computes estimates of the mean rewards $\hat{\mu}^t$ (where $\hat{\mu}^t_j$ corresponds to arm $j$) and an uncertainty vector $\epsilon^t$ (where $\epsilon^t_j$ is the uncertainty in $\hat{\mu}^t_j$). The standard UCB algorithm at time $t$ chooses arm $j_t = \arg\max_j \hat{\mu}^t_j + \epsilon^t_j$. The key idea underlying the low regret of standard UCB is that 
\begin{equation}\label{eq:UCB_standard_MAB}
    \text{[Regret at time $t$]} = O(\epsilon_{j_t}).
\end{equation}
In online fair division, instead of choosing a single arm $j_t$, the algorithm chooses a fractional allocation $X^t$. The generalization of Equation \eqref{eq:UCB_standard_MAB} to this setting is  
\begin{equation}\label{eq:to_show_regret_bound}
        \text{[Regret at time $t$]}  =  O\left(\frob{X^t}{\epsilon^t}\right).
\end{equation}
The standard UCB approach for finding a fractional allocation $X^t$ that satisfies Equation \eqref{eq:to_show_regret_bound} is to solve the optimization problem
\begin{align*}
         \arg\max_{X} \: &  \frob{X}{\mu^t_U} \\
        & \sum_{i} X_{ik} = 1 \quad \forall k   \numberthis \label{eq:lp_unconstrained}
\end{align*}
As in standard MAB, the solution to LP \eqref{eq:lp_unconstrained} will satisfy Equation \eqref{eq:to_show_regret_bound}. In online fair division, however, the algorithm must choose an $X^t$ satisfying the proportionality constraints, and the solution to LP \eqref{eq:lp_unconstrained} may not satisfy these constraints. Instead, Algorithm \ref{algo:UCB} uses LP \eqref{eq:lp_etc_finallp2} (which has the same objective function as LP \eqref{eq:lp_unconstrained}) to find an allocation $\hat{X}^t$ that satisfies the proportionality constraints with high probability. However, $\hat{X}^t$ may no longer satisfy Equation \eqref{eq:to_show_regret_bound}. This means the algorithm cannot use the allocation $\hat{X}^t$ and bound the regret with a UCB argument.

Because we cannot directly use $\hat{X}^t$, Algorithm \ref{algo:UCB} instead leverages $\hat{X}^t$ to find an allocation $X^t$ that will satisfy Equation \eqref{eq:to_show_regret_bound}. This is done using LP \eqref{eq:lp_with_slack}. LP \eqref{eq:lp_with_slack} is designed to find $Z^{ik}$ such that the $(i,k)$ entry of $Z^{ik}$ is relatively large compared to the $(i,k)$ entry of $Y^{\mu^*}$. Algorithm \ref{algo:UCB} chooses $X^t$ to be a linear combination of the $Z^{ik}$. Therefore, every entry in $X^t$ will be relatively large compared to the corresponding entry in $Y^{\mu^*}$. Furthermore, the second constraint in LP \eqref{eq:lp_with_slack} guarantees that $X^t$ will not have significantly less social welfare than $\hat{X}^t$. The bulk of the theoretical work in proving Theorem \ref{thm:sqrtT_regret} is showing that the previous two sentences together imply that $X^t$ will satisfy a (slightly more complicated) version of Equation \eqref{eq:to_show_regret_bound}. Once we show that Equation \eqref{eq:to_show_regret_bound} holds for $X^t$, a UCB argument bounds the regret of the algorithm to be $\tilde{O}(\sqrt{T})$.

\section{Properties of Proportionality Constraints}

The proof of Theorem \ref{thm:sqrtT_regret} relies on three key results about proportionality in expectation constraints, which we outline in this section.

The first result, Lemma \ref{def:notion_of_fairness}, is the reason that Algorithm \ref{algo:UCB} can explore and satisfy the proportionality in expectation constraints when the mean values are unknown.
\begin{lemma}\label{def:notion_of_fairness}
    The uniform at random allocation satisfies the proportionality in expectation constraints for any $\mu^* \in [a,b]^{n \times m}$.
\end{lemma}
\begin{proof}
    If every player is given exactly $1/n$ proportion of every item type as in the UAR allocation, then every player has value $1/n$ for their allocation. This implies that the proportionality in expectation constraints will be satisfied.
\end{proof}

The second lemma shows that the total social welfare of the optimal allocation is continuous in the mean value matrix. 

\begin{lemma}\label{def:continuous}
   Define $Y^\mu$ as the solution to LP \eqref{lp:ymu}. Then there exists a $\gamma_0$ such that if $\norm{\mu- \mu'}_1  \le \gamma_0$ for $\mu, \mu' \in  [a,b]^{n \times m}$, then $\frob{Y^{\mu}}{\mu} - \frob{Y^{\mu'}}{\mu'} \le \frac{bn}{a}\norm{\mu- \mu'}_1$.
\end{lemma}

\begin{proof}
    We provide a brief proof sketch and defer the formal proof to Appendix C. Define $\epsilon = \norm{\mu- \mu'}_1 $. Let $Z^{\mu}$ be the solution to the following modification of LP \ref{lp:ymu}: 
    \begin{align*}
        \max &  \frob{X}{\mu} \\
        \text{s.t. } & X_i \cdot \mu_i - \frac{1}{n} \ge - \epsilon  \quad \forall i \in [n]  \\
        & \sum_{i} X_{ik} = 1 \quad \forall k   \numberthis \label{eq:lp_Zmu_sketch}
    \end{align*} 
    
    We next construct a fractional allocation $W^\mu$ such that $W^{\mu}$ is a solution to LP \ref{lp:ymu} and such that $\frob{W^\mu}{\mu} - \frob{Z^{\mu}}{\mu} \ge - O(\epsilon)$. That is to say, the social welfare of $W^{\mu}$ is not much worse than that of $Z^{\mu}$. Because $W^{\mu}$ is a solution to LP \ref{lp:ymu}, we have 
    \[
        \frob{Y^{\mu}}{\mu} \ge \frob{W^{\mu}}{\mu} \ge \frob{Z^{\mu}}{\mu} - n\epsilon.
    \]
    Next, we note that $Y^{\mu'}$ is also a solution to LP \ref{eq:lp_Zmu_sketch} by construction. Therefore, it must be the case that 
    \[
        \frob{Z^{\mu}}{\mu} \ge \frob{Y^{\mu'}}{\mu}.
    \]
    
    Combining the previous two equations gives that $\frob{Y^{\mu}}{\mu} \ge \frob{Y^{\mu'}}{\mu'} - n\epsilon$. By symmetry, we also have that $\frob{Y^{\mu'}}{\mu'} \ge \frob{Y^{\mu}}{\mu} - n\epsilon$, and together with the previous equation this proves the desired result.
\end{proof}

Lemma \ref{def:fairness_to_UAR2} is the key reason that Algorithm \ref{algo:UCB} is able to find an allocation that satisfies the proportionality constraints without losing significant social welfare. Informally, this lemma states that for a known mean matrix $\mu$, there exists an allocation $X'$ such that $X'$ has only slightly less expected social welfare than $Y^{\mu}$ and such that either $X' = \mathrm{UAR}$, or $X'$ is close to $Y^{\mu}$ \textit{and} every proportionality constraint has non-negligible slack under $X'$. 
\begin{lemma}\label{def:fairness_to_UAR2}
    Define $Y^\mu$ as the solution to LP \eqref{lp:ymu}. Then for any $\gamma < \frac{a}{bn}$ and any $\mu \in [a,b]^{n \times m}$, there exists an allocation $X'$ such that  $\frob{X'}{\mu} \ge \frob{Y^\mu}{\mu} - \frac{bn\gamma}{a}$ and either $X' = \mathrm{UAR}$ or for each $i \in [n]$,
    \begin{enumerate}
        \item $X_i' \cdot \mu_i \ge  \frac{1}{n} + \gamma$ and
        \item  $\forall i\in [n], k \in [m]$, $|X'_{ik} -  Y^\mu_{ik}| \le \frac{n\gamma}{a}$.
    \end{enumerate}
\end{lemma}
The proof of Lemma \ref{def:fairness_to_UAR2} uses the same construction as in Lemma 2 of \citep{procaccia2024honor}, and we defer the proof to Appendix D. Informally, the construction either sets $X' = \mathrm{UAR}$ or constructs $X'$ by redistributing allocation away from players with large proportionality surplus (i.e. players who strictly prefer their allocation to UAR). Unlike in \citep{procaccia2024honor}, the proof uses that this redistribution process always redistributes at most $n\gamma/a$ from any (player, item) pair, thereby satisfying the second condition of Lemma \ref{def:fairness_to_UAR2}.

\section{Proof Sketch of Theorem \ref{thm:sqrtT_regret}}

We are now ready to present the proof sketch of Theorem \ref{thm:sqrtT_regret}. In Appendix A, we prove a more general result than Theorem \ref{thm:sqrtT_regret} that applies to any set of fairness constraints satisfying general versions of Properties \ref{def:notion_of_fairness}, \ref{def:continuous}, and \ref{def:fairness_to_UAR2}. For this proof sketch, we will outline the proof for proportionality constraints. See Appendix A for more details on the general version of this result.

\begin{proof}[Proof sketch]
By Lemma \ref{def:notion_of_fairness}, the UAR allocations used for the first $\sqrt{T}\log^2(T)$ steps will satisfy the proportionality constraints. The regret of the first $\sqrt{T}\log^2(T)$ steps is $\tilde{O}(\sqrt{T})$ because the regret of any one step is upper bounded by $b-a$.

Now we study what happens in the algorithm for $t \ge \sqrt{T}\log^2(T)$. Define event $E$ as the event that $\mu^* \in B(\hat{\mu}^t, \epsilon^t)$ and $\norm{\epsilon^t}_1 = \tilde{O}(T^{-1/4})$ for every $t$. By two applications of Azuma--Hoeffding's inequality, $\Pr(E) = 1-\frac{2}{3T}$. Under event $E$, for every round $t$ and for all $i,k$, the allocation $\hat{Z}^{ik}$ will satisfy the proportionality in expectation constraints for $\mu^*$ due to the first constraint in LP \eqref{eq:lp_with_slack}. Because $X^t$ is a linear combination of the $\hat{Z}^{ik}$ and the proportionality constraints are linear, this implies that under event $E$, $X^t$ will also satisfy the proportionality in expectation constraints for $\mu^*$.

Now we will bound the regret of Algorithm \ref{algo:UCB} for $t \ge \log^2(T)\sqrt{T}$. The key step in bounding this regret is showing that the regret at time $t \ge \log^2(T)\sqrt{T}$ is
\begin{equation}\label{eq:to_show_regret_at_time_t}
\frob{Y^{\mu^*}}{\mu^*} - \frob{X^t}{\mu^*} =  \tilde{O}\left(\frob{X^t}{\epsilon^t}+ \frac{1}{\sqrt{T}}\right) .
\end{equation}
In order to show Equation \eqref{eq:to_show_regret_at_time_t}, we first bound the regret at time $t$ by an expression which does not contain any terms involving $\mu^*$. 

Under event $E$, $\mu^* \in B(\hat{\mu}^t, \epsilon^t)$. $G^t$ forms a grid on $B(\hat{\mu}^t, \epsilon^t)$, and therefore there exists some element $\mu_g \in G^t$ such that $\norm{\mu^* - \mu_g}_{\infty} \le \frac{1}{\sqrt{T}}$. By Lemma \ref{def:continuous}, this implies that 
\begin{equation}\label{eq:to_use1}
|\frob{Y^{\mu^*}}{\mu^*} - \frob{Y^{\mu_g}}{\mu_g}| = O\left(\frac{1}{\sqrt{T}}\right).
\end{equation}
Using the first constraint of LP \eqref{eq:lp_with_slack}, we can bound $\frob{\hat{X}^t}{\mu^U_t} - \frob{\hat{Z}^{ik}}{\mu^U_t}$. By construction of $X^t$ and $\mu^g$, this implies that
\begin{align*}
&\frob{\hat{X}^t}{\mu^g} - \frob{X^t}{\mu^*} \\
&= \tilde{O}_T\left( \max_{\mu \in G^t} \frob{Y^{\mu}}{\epsilon^t} + \frob{\hat{X}^t}{\epsilon_t} + \frob{\hat{Z}^{ik}}{\epsilon^t}\right).\numberthis \label{eq:to_use2}
\end{align*}
Furthermore, a series of algebraic steps with a UCB argument shows that 
\begin{align*}
    &\frob{Y^{\mu^g}}{\mu^g} - \frob{\hat{X}^t}{\mu^g} \\
    &=  \tilde{O}\left(\max_{\mu \in G^t} \frob{Y^{\mu}}{\epsilon^t} + \langle \hat{X}^t, \epsilon^t \rangle_F   \right). \numberthis \label{eq:to_use3}
\end{align*}
Combining Equations \eqref{eq:to_use1}, \eqref{eq:to_use2}, and \eqref{eq:to_use3} gives the following key result.
\begin{align*}
&\frob{Y^{\mu^*}}{\mu^*} - \frob{X^t}{\mu^*}\\
&=  \tilde{O}\left( \frob{\hat{X}^t}{\epsilon^t} + \max_{\mu \in G^t} \frob{Y^{\mu}}{\epsilon^t} + \frob{X^t}{\epsilon^t} +\frac{1}{\sqrt{T}}\right). \numberthis \label{eq:to_bound_terms}
\end{align*}
To show Equation \eqref{eq:to_show_regret_at_time_t} from Equation \eqref{eq:to_bound_terms}, we bound the first two  terms in Equation \eqref{eq:to_bound_terms} by $O(\frob{X^t}{\epsilon^t} + \frac{1}{\sqrt{T}})$ conditional on event $E$. First, we show that for all $i$ and $k$, $\hat{X}^t$ satisfies the constraints of LP \eqref{eq:lp_with_slack}. We can then conclude that $\hat{Z}^{ik}_{ik} \ge \hat{X}^t_{ik}$ because $\hat{Z}^{ik}$ is the solution to LP \eqref{eq:lp_with_slack}. $X^t$ is a linear combination of the $\hat{Z}^{ik}$, and therefore the previous sentence implies that $X^t_{ik} \ge \frac{1}{nm}\hat{X}^t_{ik}$. This implies that 
\begin{equation}\label{eq:to_combine1}
    \frob{\hat{X}^t}{\epsilon^t} = O(\frob{X^t}{\epsilon^t}).
\end{equation}
Under event $E$, Lemma \ref{def:fairness_to_UAR2} with a carefully chosen value of $\gamma$ implies that for every $\mu \in G^t$, there exists an allocation $X^{\mu}$ such that $\frob{X^{\mu}}{\mu}$ is similar to $\frob{Y^{\mu}}{\mu}$ and such that $X^{\mu}$ is a solution to LP \ref{eq:lp_with_slack}. The fact that $X^{\mu}$ is a solution to LP \ref{eq:lp_with_slack} is a result of the careful construction of constraints in LP \ref{eq:lp_with_slack} and choice of $\gamma$ for Lemma \ref{def:fairness_to_UAR2}. Because $X^\mu$ is a solution to LP \eqref{eq:lp_with_slack}, by the same logic as in the previous paragraph, we have that $\frob{X^{\mu}}{\epsilon^t} = O(\frob{X^t}{\epsilon^t})$. $X^{\mu}$ was constructed in Lemma \ref{def:fairness_to_UAR2} such that the elements of $X^{\mu}$ are close to the elements of $Y^{\mu}$. Therefore, $\frob{Y^{\mu}}{\epsilon^t} = O(\frob{X^t}{\epsilon^t} + \frac{1}{\sqrt{T}})$ for all $\mu \in G^t$, or equivalently 
\begin{align*}
    \max_{\mu \in G^t} \frob{Y^{\mu}}{\epsilon^t} = O(\frob{X^t}{\epsilon^t} + \frac{1}{\sqrt{T}}). \numberthis \label{eq:to_combine2}
\end{align*}
Putting together Equations \eqref{eq:to_bound_terms}, \eqref{eq:to_combine1}, and \eqref{eq:to_combine2} gives the desired result of Equation \eqref{eq:to_show_regret_at_time_t}. Using Equation \eqref{eq:to_show_regret_at_time_t}, we can conclude with an upper confidence bound argument to upper bound the total $T$-step regret by
\begin{align*}
&\sum_{t=0}^{T-1} \frob{Y^{\mu^*}}{\mu^*} - \frob{X^t}{\mu^*} \\
&\le \tilde{O}(\sqrt{T}) + \sum_{t=\log^2(T)\sqrt{T}}^{T-1} \tilde{O}\left(\frob{X^t}{\epsilon^t} + \frac{1}{\sqrt{T}}\right)\\
&= \tilde{O}(\sqrt{T}).
\end{align*}
\end{proof}

\subsection{Lower Bounds}

The regret of $\tilde{O}(\sqrt{T})$ in Theorem \ref{thm:sqrtT_regret} is tight for $T$ up to log factors because $\tilde{\Omega}(\sqrt{T})$ is the standard lower bound for stochastic multi-armed bandit problems. A natural follow-up question to Theorem \ref{thm:sqrtT_regret} is whether an equivalent result holds when the algorithm must satisfy other notions of fairness such as envy-freeness in expectation. For the online fair division problem, envy-freeness in expectation can be represented as constraints in the following way:

\begin{definition}\label{def:ef}
    An algorithm $\alg$ that uses fractional allocation $X^t$ at time $t$ satisfies \emph{envy-freeness in expectation  }if for all $t < T$ and all $i \in [n]$,
    \[(X_t)_i \cdot \mu^*_i \ge \max_{i' \in [n]} \: \: (X_t)_{i'} \cdot \mu^*_i \quad \forall i \in [n].
    \]
\end{definition}
Theorem \ref{thm:lower_bounds} shows that an equivalent result to Theorem \ref{thm:sqrtT_regret} does not hold for envy-freeness in expectation, and in fact the best regret possible while maintaining envy-freeness in expectation is $\tilde{\Omega}(T^{2/3})$.

\begin{thm}\label{thm:lower_bounds}
    There exists $a,b,n,m$ such that no algorithm can, for all $\mu^* \in [a,b]^{n \times m}$ with rows that add to $1$, both satisfy the envy-freeness in expectation constraints and achieve regret of less than $\frac{T^{2/3}}{\log(T)}$ with probability at least $1-1/T$.
\end{thm}
\begin{proof}[Proof sketch]
    The proof of this theorem extends the lower bound construction of \citep{procaccia2024honor} to an example with normalized mean values. We prove the desired result by contradiction. First, we assume that an algorithm $\alg$ does achieve regret of less than $\frac{T^{2/3}}{\log(T)}$ and satisfies the envy-freeness in expectation constraints with probability greater than $1-1/T$ for all mean value matrices with normalized values. We then construct two mean value matrices each with three players and three types of items that differ only in two entries by $T^{-1/3}$, and therefore are difficult to distinguish between. We then show that $\alg$ cannot simultaneously have low regret and satisfy the envy-freeness in expectation constraints for both of these mean value matrices, which leads to a contradiction. See Appendix E for the formal proof.
\end{proof}

\section{Discussion}\label{sec:discussion}

\subsection{Non-random Item Types}
We studied the online fair division problem where $T$ items arrive online and each item's type is drawn uniformly at random. In this section, we discuss how our results extend to a similar problem where all item types are deterministic. Suppose that instead of a single item arriving at each of $T$ time steps, a basket containing exactly one of every item type arrives at each of $T/m$ time steps. The algorithm then must allocate every item in the basket among the $n$ players using a random fractional allocation. The goal in this alternative setting is to maximize expected social welfare while satisfying the proportionality in expectation constraints for each basket. 

For the setting studied in this paper, there are three sources of randomness: random item types, random player values, and random fractional allocations. In this alternative setting the item types are not random, and therefore there are only two sources of randomness: random player values and random fractional allocations. Note that $m$ items with uniform at random types are equal in expectation to a basket of $m$ items with one of every item type. Furthermore, both the constraints and the reward function for both settings are expressed as expectations. Therefore, the difference in sources of randomness does not fundamentally change the problem. This implies that the results from this paper (such as Theorem \ref{thm:sqrtT_regret}) carry over to this setting with deterministic item types. 

\subsection{Limitations and Future Directions}

In this section, we discuss the limitations of this paper and some potential future directions. One limitation of Algorithm \ref{algo:UCB} is that the runtime is not linear due to the second round of optimization. An open question is whether an algorithm with linear runtime can achieve the same rate of regret. Furthermore, the main result of this paper, Theorem \ref{thm:sqrtT_regret}, only applies to proportionality constraints. In the proof of Theorem \ref{thm:sqrtT_regret}, we do present a general algorithm that achieves $\tilde{O}(\sqrt{T})$ regret for any fairness constraints satisfying certain properties. An open question is whether there exist other types of fairness constraints (e.g. equitability) for which $\tilde{O}(\sqrt{T})$ is also possible. We leave this as an open question for future work.

Another interesting question is whether the ideas in this paper can be extended to other fair division problems. More broadly, the techniques used in this paper are not exclusive to fairness constraints. Therefore, another question is whether similar ideas can lead to algorithms that achieve $\tilde{O}(\sqrt{T})$ regret under even more general classes of constraints.

\section{Acknowledgements}

Schiffer was supported by an NSF Graduate Research Fellowship. Zhang was supported by an NSF Graduate Research Fellowship. The authors would also like to thank Ariel Procaccia for helpful discussions.

%% file: appendix.tex
\section{Proof of Theorem \ref{thm:sqrtT_regret}}\label{app:general_definitions}

To prove Theorem \ref{thm:sqrtT_regret}, we use a general framework for analyzing fairness constraints introduced in \cite{procaccia2024honor}. Note that Property \ref{def:continuous_app} is a new property not previously studied in \cite{procaccia2024honor}, while Properties \ref{def:fairness_to_UAR2_app} and \ref{def:lipschitz} are variants of properties from \cite{procaccia2024honor}. We also note that the usage of the properties in the proof of Lemma \ref{lemma:general_sqrtT_regret} is significantly different than how these properties are used in \cite{procaccia2024honor}.

We consider sets of fairness constraints $\{(B_\ell(\mu^*) , c_\ell)\}_{\ell=1}^{L}$, where $B_\ell(\mu^*) \in \mathbb{R}^{n \times m}$ and $c_\ell \in \mathbb{R}$. An algorithm $\alg$ that uses allocation $X^t$ at time $t$ satisfies these constraints if $\frob{B_\ell(\mu^*)}{X^t} \ge c_\ell$ for all $t$. The expected social welfare maximizing allocation satisfying these constraints for a given value matrix $\mu$ is
\begin{align*}
    Y^\mu := \arg\max_X & \: \:  \frob{X}{\mu} \\
    \text{s.t. } & \frob{B_\ell(\mu)}{X} \ge c_\ell \quad \forall \ell \\
    & \sum_{i} X_{ik} = 1 \quad \forall k  \numberthis \label{lp:ymu_app}
\end{align*}
The regret of $\alg$ (that uses allocation $X^t$ at time $t$) for constraints $\{(B_\ell(\mu^*), c_\ell\}_{\ell=1}^{L}$ can be represented as $T \cdot \frob{Y^{\mu^*}}{\mu^*} - \sum_{t=0}^{T-1} \frob{X^t}{\mu^*}$.

In this framework, we can define the proportionality in expectation constraints as follows:
 
 \begin{remark}\label{remark_2}
    For every $\ell \in [n]$, construct $B_\ell$ as follows. For every $k \in [m]$, $B_{\ell}(\mu^*)_{\ell k} = \mu^*_{\ell k}$ and $B_{\ell}(\mu^*)_{ik} = 0$ for every $i \ne \ell$.  The proportionality in expectation constraints can be written as $\left\{\left(B_\ell(\mu^*),\frac{1}{n}\right) \right\}_{\ell=1}^{n}$.
\end{remark}

We now state the general versions of Properties \ref{def:notion_of_fairness}, \ref{def:continuous}, and \ref{def:fairness_to_UAR2}.

\begin{property}\label{def:notion_of_fairness_app}
    A constraint $(B, c)$ is  \emph{satisfied by UAR} if  the uniform at random (UAR) allocation satisfies the constraint, i.e. $\frac{1}{n}\norm{B}_1 \le c$.
\end{property}

\begin{property}\label{def:fairness_to_UAR2_app}
    Define $Y^\mu$ as the solution to LP \eqref{lp:ymu_app}. Then there exists $ \gamma_0, C_{\mathrm{P}\ref{def:fairness_to_UAR2_app}} > 0$ such that for any $\gamma < \gamma_0$ and any $\mu \in [a,b]^{n \times m}$, there exists $X'$ such that $\frob{X'}{\mu} \ge \frob{Y^\mu}{\mu} - C_{\mathrm{P}\ref{def:fairness_to_UAR2_app}}\gamma$ and either $X' = \mathrm{UAR}$ or
    \begin{enumerate}
        \item $\langle B_\ell(\mu), X' \rangle_F \ge  c_\ell + \gamma$ for all $\ell \in [L]$ and
        \item $\forall i\in [n], k \in [m]$, $|X'_{ik} -  Y^\mu_{ik}| \le C_{\mathrm{P}\ref{def:fairness_to_UAR2_app}}\gamma$.
    \end{enumerate}
\end{property}

\begin{property}\label{def:continuous_app}
   Define $Y^{\mu}$ as the solution to LP \eqref{lp:ymu_app}. There exists $\gamma_0,C_{\mathrm{P}\ref{def:continuous_app}}$ such that if $\norm{\mu- \mu'}_1  \le \gamma_0$ for $\mu, \mu' \in  [a,b]^{n \times m}$, then $\frob{Y^{\mu}}{\mu} - \frob{Y^{\mu'}}{\mu'} \le C_{\mathrm{P}\ref{def:continuous_app}} \norm{\mu- \mu'}_1$.
\end{property}

We also need one additional technical property.
 \begin{property}\label{def:lipschitz}
     There exists a constant $K > 0$ such that for all $\mu^1, \mu^2 \in [a,b]^{n \times m}$, if $|\mu^1_{ik} - \mu^2_{ik}| \le \epsilon_{ik}$ then $|B_\ell(\mu^1)_{ik} - B_\ell(\mu^2)_{ik}| \le K\epsilon_{ik}$ for all $i,k$.
 \end{property}

\begin{observation}\label{observation:satisfies_lipschitz_and_zeros}
    The proportionality in expectation constraints defined in Remark \ref{remark_2} satisfies Property \ref{def:lipschitz} with $K = 1$.
\end{observation}
Equipped with this additional property, we state the more general form of Algorithm \ref{algo:UCB} and Theorem \ref{thm:sqrtT_regret} in Algorithm \ref{algo:UCB_app} and Lemma \ref{lemma:general_sqrtT_regret}.

\begin{lemma}\label{lemma:general_sqrtT_regret}
    Suppose $\left\{\{(B_\ell(\mu), c_\ell)\}_{\ell=1}^{L}\right\}_{\mu \in [a,b]^{n \times m}}$ satisfy Properties \ref{def:notion_of_fairness_app}, \ref{def:fairness_to_UAR2_app}, \ref{def:continuous_app}, and \ref{def:lipschitz}.  Then with probability $1-1/T$, Algorithm \ref{algo:UCB} satisfies the constraints $\{(B_\ell(\mu^*), c_\ell)\}_{\ell=1}^{L}$ and has regret of $\tilde{O}(KC_{\mathrm{P}\ref{def:continuous_app}} C_{\mathrm{P}\ref{def:fairness_to_UAR2_app}}n^3m^3\sqrt{T})$.
\end{lemma}

 \begin{algorithm}[htb!]
 \caption{Fair Upper-Confidence-Bound}
 \label{algo:UCB_app}
 \begin{algorithmic}
\Require $n,m,T,a,b, \{\{(B_\ell(\mu), c_\ell)\}_{\ell=1}^{L}\}_{\mu}$
\State $C_{\mathrm{P}\ref{def:fairness_to_UAR2_app}} \gets \text{From Prop \ref{def:fairness_to_UAR2_app}}$, $K \gets \text{From Prop \ref{def:lipschitz}}$ \label{line:CK}
\For{$t \gets 0$ to $\log^2(T)\sqrt{T}-1$}
    \State Use $X^t = \mathrm{UAR}$.
\EndFor
\For{$t \gets \log^2(T)\sqrt{T}$ to $T$}
     \State $N^t_{ik} \gets  \sum_{\tau=0}^{t-1} 1_{k_\tau = k, i_\tau = i}$
     \State $\hat{\mu}^t_{ik} \gets  \frac{1}{N^t_{ik}} \sum_{\tau = 0}^{t-1} 1_{k_\tau = k, i_\tau = i} V_{i_\tau}(j_\tau)$
     \State $\epsilon^t_{ik} = \sqrt{\log^2(6nmT)/(N^t_{ik})}$
     \State $(\mu_U^t)_{ik} = \hat{\mu}_{ik}^t + \epsilon_{ik}^t$
     \State $G^t = \left\{\mu \in B(\hat{\mu}^t, \epsilon^t): \sqrt{T}\mu_{ik} \in \mathbb{Z} 
\quad \forall i,k\right\}$
     \State $\hat{X}^t \gets$ Solution to
             \begin{align*}
         \max_{X} \: &  \frob{X}{\mu^t_U} \\
        \text{s.t. } & \langle B_\ell(\mu), X \rangle_F \ge c_\ell  \quad \forall \ell \in [L], \forall \mu \in B(\hat{\mu^t},\epsilon^t) \\
        & \sum_{i} X_{ik} = 1 \quad \forall k   \numberthis \label{eq:lp_etc_finallp2_app}
        \end{align*}
    \State $\forall i \in [n], \forall k \in [m]$, $\hat{Z}^{ik} \gets$ Solution to the following LP:
    \begin{align*}
        \max \: &  X_{ik} \\
        \text{s.t. } & \langle B_\ell(\mu), X \rangle_F \ge c_\ell  \quad \forall \ell \in [L], \forall \mu \in B(\hat{\mu}^t,\epsilon^t) \\
        & \frob{X}{\mu^t_U} \ge \frob{\hat{X}^t}{\mu^t_U}  -   4KC_{\mathrm{P}\ref{def:fairness_to_UAR2_app}}\max_{\mu \in G^t} \frob{Y^{\mu}}{\epsilon^t}  - 2\frob{\hat{X}^t}{\epsilon^t}\\
        & \sum_{i'} X_{i'k'} = 1 \quad \forall k'  \numberthis \label{eq:lp_with_slack_app}
    \end{align*}
    \State Use $X^t = \frac{1}{nm} \left(\sum_{i,k} \hat{Z}^{ik} \right)$ 
\EndFor
 \end{algorithmic}
 \end{algorithm}

\newpage

Now we can use Lemma \ref{lemma:general_sqrtT_regret} to prove Theorem \ref{thm:sqrtT_regret}. 
\begin{proof}[Proof of Theorem \ref{thm:sqrtT_regret}]
    By Observation \ref{observation:satisfies_lipschitz_and_zeros}, the proportionality in expectation constraints defined in Remark \ref{remark_2} satisfies Property \ref{def:lipschitz}. Furthermore, by Lemmas \ref{def:notion_of_fairness}, \ref{def:continuous}, and \ref{def:fairness_to_UAR2}, the proportionality in expectation constraints satisfy Properties \ref{def:notion_of_fairness_app}, \ref{def:fairness_to_UAR2_app} (with $C_{\mathrm{P}\ref{def:fairness_to_UAR2_app}} = \frac{bn}{a}$),  and \ref{def:continuous_app} (with $C_{\mathrm{P}\ref{def:continuous_app}} = n)$. Therefore, Lemma \ref{lemma:general_sqrtT_regret} directly proves Theorem \ref{thm:sqrtT_regret}.
\end{proof}

\section{Proof of Lemma \ref{lemma:general_sqrtT_regret}}

\begin{proof}
    First, we bound the regret and prove that the constraints are satisfied for the first $\log^2(T)\sqrt{T}$ steps. The regret of any single step is at most $b-a$ due to the assumptions of bounds on the mean values. Therefore, the regret of the first $\log^2(T)\sqrt{T}$ steps can be bounded by $(b-a)\log^2(T)\sqrt{T} = \tilde{O}(\sqrt{T})$. Furthermore, by Property \ref{def:notion_of_fairness_app}, the fractional allocations used in these steps satisfy the constraints because they are uniform at random allocations. 
    
    For the rest of the proof, we will consider $t \ge \log^2(T)\sqrt{T}$. First, we construct a high probability event $E$ such that under event $E$, the uncertainties $\epsilon^t$ are small and $\mu^* \in B(\hat{\mu}^t, \epsilon^t)$.

    Define event $E$ as
    \[
       E := \bigcap_{t \ge \log^2(T)\sqrt{T}} \left\{\norm{\epsilon^t}_1 \le nm\sqrt{\frac{2nm\log^2(6nmT)}{\sqrt{T}}} \right\} \cap  \left\{ \forall i \in [n],k \in [m], |\hat{\mu}_{ik}^t - \mu^*_{ik}| \le \epsilon_{ik}^t\right\}.
    \]
    We can bound the probability of $E$ using the following lemma.
    \begin{lemma}\label{lemma:prob_of_event_E}
        Using the notation above, $\Pr(E) \ge 1 - \frac{2}{3T}$.
    \end{lemma}
    The proof of Lemma \ref{lemma:prob_of_event_E} can be found in Appendix \ref{app:prob_of_event_E}.
    
     The bulk of the rest of this proof will focus on bounding the regret at every time $t$ after the warm-up period conditional on event $E$. We then conclude by showing that $X^t$ satisfies the fairness constraints for all $t$ conditional on event $E$. 
    
   For every $\mu \in G^t$, we will now define an allocation $X^{\mu}$ using Property \ref{def:fairness_to_UAR2_app} with a specific choice of $\gamma$. This $X^{\mu}$ will be important for bounding the regret at time $t$. Consider any $\mu \in G^t$. Under event $E$ and for sufficiently large $T$, 
   \[
   \gamma := 4K\frob{Y^\mu}{\epsilon^t} \le   4K\norm{\epsilon^t}_1 = \tilde{O}(Kn^{1.5}m^{1.5}T^{-1/4}) \le \gamma_0,
   \]
   where $K$ is the constant from Property \ref{def:lipschitz} and $\gamma_0$ is from Property \ref{def:fairness_to_UAR2_app}. By Property \ref{def:fairness_to_UAR2_app}, there exists a fractional allocation $X^\mu$ such that 
 \begin{equation}\label{eq:swclose}
        |\frob{X^{\mu}}{\mu} - \frob{Y^{\mu}}{\mu}| \le  C_{\mathrm{P}\ref{def:fairness_to_UAR2_app}}4K\frob{Y^{\mu}}{\epsilon^t}.
    \end{equation}
    Furthermore, either $X^{\mu} = \mathrm{UAR}$ or both 
    \begin{equation}\label{eq:first_caseB23case}
            \frob{B_\ell(\mu)}{X^\mu} \ge  c_\ell + 4K\frob{Y^\mu}{\epsilon^t}.
    \end{equation}
    and
    \begin{equation}\label{eq:close_approx}
        |X_{ik}^\mu - Y_{ik}^{\mu}| \le C_{\mathrm{P}\ref{def:fairness_to_UAR2_app}}4K\frob{Y^{\mu}}{\epsilon^t} \le O(KC_{\mathrm{P}\ref{def:fairness_to_UAR2_app}}\norm{\epsilon^t}_1)
    \end{equation} 
    We can also show that for every $\mu \in GYt$, $X^\mu$ satisfies the constraints of the linear programs in Algorithm \ref{algo:UCB}.
    \begin{lemma}\label{lemma:Xmu}
        Under event $E$, for every $\nu \in G^t$, $X^{\nu}$\,---\,as defined above in Equations \eqref{eq:swclose}--\eqref{eq:close_approx}\,---\, satisfies the constraints of LP \eqref{eq:lp_etc_finallp2_app} and LP \eqref{eq:lp_with_slack_app}.
    \end{lemma}
    The proof of Lemma \ref{lemma:Xmu} can be found in Appendix \ref{app:Xmu}.
  
    For the next part of the proof, we will bound the regret at time $t$ by quantities that do not depend on $\mu^*$. Define $\mu^g$ as the element in $G^t$ that is closest to $\mu^*$. Because $\mu^* \in B(\hat{\mu}^t, \epsilon^t)$ under event $E$ and because $G^t$ forms an evenly spaced grid on $B(\hat{\mu}^t,\epsilon^t)$, we have that under event $E$,
     \begin{equation}\label{eq:closeg}
         \norm{\mu^*-\mu^g}_1 \le \frac{nm}{\sqrt{T}}.
     \end{equation}    
    When $\mu^* \in B(\mu^t, \epsilon^t)$, for all $i,k$ we have $\mu^*_{ik} \ge \hat{\mu}^t_{ik} - \epsilon^t_{ik}$. Because $(\mu_U^t)_{ik} = \hat{\mu}^t_{ik} + \epsilon^t_{ik}$, this implies that $\mu^*_{ik} \ge (\mu^t_U)_{ik} - 2\epsilon^t_{ik}$. Using this in the first line below, we have that
  \begin{align*}
      \frob{\hat{Z}^{ik}}{\mu^*} &\ge \frob{\hat{Z}^{ik}}{\mu^t_U- 2 \epsilon^t} && \text{[Event $E$]}\\
      &= \frob{\hat{Z}^{ik}}{\mu^t_U} - 2\frob{\hat{Z}^{ik}}{\epsilon^t} &&\\
      &\ge \frob{\hat{X}^t}{\mu_U^t} -   4KC_{\mathrm{P}\ref{def:fairness_to_UAR2_app}} \cdot  \max_{\mu \in G^t} \frob{Y^{\mu}}{\epsilon^t} - 2\frob{\hat{X}^t}{\epsilon^t} - 2\frob{\hat{Z}^{ik}}{\epsilon^t}  && \text{[LP \eqref{eq:lp_with_slack_app}]}\\
      &\ge \frob{\hat{X}^t}{\mu^g} -  4KC_{\mathrm{P}\ref{def:fairness_to_UAR2_app}} \cdot \max_{\mu \in G^t} \frob{Y^{\mu}}{\epsilon^t} - 2\frob{\hat{X}^t}{\epsilon^t} - 2\frob{\hat{Z}^{ik}}{\epsilon^t}.  && \text{[$(\mu_U^t)_{ik} \ge \mu^g_{ik}$]} \numberthis \label{eq:Zik}
  \end{align*}

 We can now bound the regret at time $t\ge \log^2(T)\sqrt{T}$ as follows:
  {\footnotesize
  \begin{align*}
      &\frob{Y^{\mu^*}}{\mu^*} - \frob{X^t}{\mu^*}\\
      &= \frob{Y^{\mu^*}}{\mu^*} - \frac{1}{nm} \left(  \sum_{i,k} \frob{\hat{Z}^{ik}}{\mu^*} \right) && \text{[Def of $X^t$]}\\
      &\le \frob{Y^{\mu^g}}{\mu^g} + C_{\mathrm{P}\ref{def:continuous_app}} \norm{\mu^*-\mu^g}_1 - \frac{1}{nm} \left(  \sum_{i,k} \frob{\hat{Z}^{ik}}{\mu^*} \right) && \text{[Prop \ref{def:continuous_app}]} \\
      &\le \frob{Y^{\mu^g}}{\mu^g} + \frac{C_{\mathrm{P}\ref{def:continuous_app}}nm}{\sqrt{T}} - \frac{1}{nm} \left(  \sum_{i,k} \frob{\hat{Z}^{ik}}{\mu^*} \right) && \text{[ Eq. \eqref{eq:closeg}]} \\
      &\le \frob{Y^{\mu^g}}{\mu^g}+ \frac{C_{\mathrm{P}\ref{def:continuous_app}}nm}{\sqrt{T}} -  \frac{1}{nm}\sum_{i,k} \left(\frob{\hat{X}^t}{\mu^g} -  4KC_{\mathrm{P}\ref{def:fairness_to_UAR2_app}} \cdot \max_{\mu \in G^t} \frob{Y^{\mu}}{\epsilon^t} - 2\frob{\hat{X}^t}{\epsilon^t} - 2\frob{\hat{Z}^{ik}}{\epsilon^t}\right)  && \text{[Eq. \eqref{eq:Zik}]} \\
      &\le \frob{Y^{\mu^g}}{\mu^g} - \frob{\hat{X}^t}{\mu^g} +   4KC_{\mathrm{P}\ref{def:fairness_to_UAR2_app}} \max_{\mu \in G^t} \frob{Y^{\mu}}{\epsilon^t} + 2\frob{\hat{X}^t}{\epsilon^t}+  \frac{2}{nm}\sum_{i,k} \frob{\hat{Z}^{ik}}{\epsilon^t} + \frac{C_{\mathrm{P}\ref{def:continuous_app}}nm}{\sqrt{T}}  \\
      &= \frob{Y^{\mu^g}}{\mu^g} - \frob{\hat{X}^t}{\mu^g} +  4KC_{\mathrm{P}\ref{def:fairness_to_UAR2_app}}\max_{\mu \in G^t} \frob{Y^{\mu}}{\epsilon^t}+ 2\frob{\hat{X}^t}{\epsilon^t}+  2\frob{X^t}{\epsilon^t}+ \frac{C_{\mathrm{P}\ref{def:continuous_app}}nm}{\sqrt{T}}. && \text{[Def of $X^t$]}  \numberthis \label{eq:bound_Ymu1}
  \end{align*}}
 This gives us an upper bound on the regret at time $t$ that does not depend on $\mu^*$.
 
  Next, we focus on the term $ \frob{Y^{\mu^g}}{\mu^g} - \frob{\hat{X}^t}{\mu^g}$. By Lemma \ref{lemma:Xmu}, $X^{\mu^g}$  satisfies the constraints of LP \eqref{eq:lp_etc_finallp2_app}. Therefore,
    \begin{equation}\label{eq:ucb1}
        \frob{X^{\mu^g}}{\mu_U^t} \le \frob{\hat{X}^t}{\mu_U^t}.
    \end{equation}
    We also need the following inequality:
    {\small
    \begin{align*}
      &\frob{X^{\mu^g}}{\mu^g} -    \frob{\hat{X}^t}{\mu^g}\\
      &= \frob{X^{\mu^g}}{\mu^t_U} -    \frob{\hat{X}^t}{\mu^t_U}  + \left(\frob{X^{\mu^g}}{\mu^g} - \frob{X^{\mu^g}}{\mu^t_U} \right) -    \left(\frob{\hat{X}^t}{\mu^g} - \frob{\hat{X}^t}{\mu^t_U} \right)  \\
      &\le \left(\frob{X^{\mu^g}}{\mu^g} - \frob{X^{\mu^g}}{\mu^t_U} \right) -    \left(\frob{\hat{X}^t}{\mu^g} - \frob{\hat{X}^t}{\mu^t_U} \right) && \text{[Eq. \eqref{eq:ucb1}]}\\
      &\le  -    \left(\frob{\hat{X}^t}{\mu^g} - \frob{\hat{X}^t}{\mu^t_U} \right) && \text{[$\mu^g_{ik} \le (\mu_U^t)_{ik}, \: \forall i,k$]}\\
      &=  \frob{\hat{X}^t}{\mu_U^t - \mu^g}\\
      &\le 2\frob{\hat{X}^t}{\epsilon^t}. \numberthis \label{eq:comp_mu_g}
    \end{align*}
    }
      Therefore, we have that
    \begin{align*}
        \frob{Y^{\mu^g}}{\mu^g} - \frob{\hat{X}^t}{\mu^g} &= \frob{Y^{\mu^g}}{\mu^g} -\frob{X^{\mu^g}}{\mu^g} + \frob{X^{\mu^g} }{\mu^g}-  \frob{\hat{X}^t}{\mu^g} \\
        &\le O\left(C_{\mathrm{P}\ref{def:fairness_to_UAR2_app}}K\frob{Y^{\mu^g}}{\epsilon^t} + \langle \hat{X}^t, \epsilon^t \rangle_F \right) && \text{[Equations \eqref{eq:swclose}, \eqref{eq:comp_mu_g}]}\\
        &\le O\left(C_{\mathrm{P}\ref{def:fairness_to_UAR2_app}}K\max_{\mu \in G^t} \frob{Y^{\mu}}{\epsilon^t} + \langle \hat{X}^t, \epsilon^t \rangle_F   \right).  \numberthis \label{eq:bound_Ymu2}
    \end{align*}
    Combining Equations \eqref{eq:bound_Ymu1} and \eqref{eq:bound_Ymu2}, we have that
    \begin{align*}
         \frob{Y^{\mu^*}}{\mu^*} - \frob{X^t}{\mu^*} \le \tilde{O}\left( C_{\mathrm{P}\ref{def:fairness_to_UAR2_app}}K\max_{\mu \in G^t} \frob{Y^{\mu}}{\epsilon^t} + \langle \hat{X}^t, \epsilon^t \rangle_F + \frob{X^t}{\epsilon^t} + \frac{C_{\mathrm{P}\ref{def:continuous_app}}nm}{\sqrt{T}}  \right). \numberthis \label{eq:diff_in_util}
    \end{align*}
    Now we will bound each of the first two terms in Equation \eqref{eq:diff_in_util}. By Lemma \ref{lemma:Xmu}, for all $\mu \in G^t$,  $X^\mu$ satisfies the constraints of LP \eqref{eq:lp_with_slack_app}, and by construction $\hat{X}^t$ also satisfies the constraints of LP \eqref{eq:lp_with_slack_app}. Therefore, we must have that for all $i,k$ and all $\mu \in G^t$,
    \begin{equation}\label{eq:X_to_Z_ind}
        X^\mu_{ik}  \le \hat{Z}^{ik}_{ik} \le nm \cdot X_{ik}^t
    \end{equation}
    and
    \begin{equation}\label{eq:X_to_Z_ind2}
        \hat{X}^t_{ik} \le \hat{Z}^{ik}_{ik} \le nm \cdot X_{ik}^t.
    \end{equation}
       Equation \eqref{eq:X_to_Z_ind} implies that for every $\mu \in G^t$,
    \begin{equation}\label{eq:X_to_Z}
    \frob{X^\mu}{\epsilon^t} \le nm \cdot \frob{X^t}{\epsilon^t}.
    \end{equation} 
    
    If there exists $\mu \in G^t$ such that $X^{\mu} = \mathrm{UAR}$, then Equation \eqref{eq:X_to_Z_ind} implies that $X_{ik}^t \ge \frac{1}{mn^2}$ for all $i,k$. This implies that 
    \begin{equation}\label{eq:one_mu_is_UAR}
        \max_{\mu \in G^t} \frob{Y^{\mu}}{\epsilon^t} \le \norm{\epsilon^t}_1 \le mn^2\frob{X^t}{\epsilon^t}.
    \end{equation}
 
    On the other hand, if there does not exist $\mu \in G^t$ such that $X^{\mu} = \mathrm{UAR}$, then for all $\mu \in G^t$,
    \begin{align*}
        \max_{\mu \in G^t} \frob{Y^{\mu}}{\epsilon^t} &\le  \max_{\mu \in G^t} \frob{X^{\mu}}{\epsilon^t} +  \max_{\mu \in G^t} (\frob{Y^{\mu}}{\epsilon^t} -  \frob{X^{\mu}}{\epsilon^t} )\\
        &= \max_{\mu \in G^t} \frob{X^{\mu}}{\epsilon^t} +\max_{\mu \in G^t} \sum_{i,k} |Y^\mu_{ik} - X^\mu_{ik}|\epsilon^t_{ik}\\
        &\le \max_{\mu \in G^t} \frob{X^{\mu}}{\epsilon^t} +\max_{\mu \in G^t} \sum_{i,k} O(KC_{\mathrm{P}\ref{def:fairness_to_UAR2_app}}\norm{\epsilon^t}_1)\epsilon^t_{ik} && \text{[Equation \eqref{eq:close_approx}]}\\
        &\le \max_{\mu \in G^t} \frob{X^{\mu}}{\epsilon^t} + O(KC_{\mathrm{P}\ref{def:fairness_to_UAR2_app}}\norm{\epsilon^t}_1^2) \\
        &\le nm \cdot \frob{X^t}{\epsilon^t} + \tilde{O}(KC_{\mathrm{P}\ref{def:fairness_to_UAR2_app}}n^3m^3/\sqrt{T}). && \text{[Equation \eqref{eq:X_to_Z}, Event $E$]} \numberthis \label{eq:max_upper_bound}
    \end{align*}
    Combining Equations \eqref{eq:one_mu_is_UAR} and \eqref{eq:max_upper_bound}, we have that under event $E$,
    \begin{equation}\label{eq:bound_on_ymu_in_both_cases}
        \max_{\mu \in G^t} \frob{Y^{\mu}}{\epsilon^t} \le \tilde{O} \left(mn^2\frob{X^t}{\epsilon^t} + \frac{KC_{\mathrm{P}\ref{def:continuous_app}} C_{\mathrm{P}\ref{def:fairness_to_UAR2_app}}n^3m^3}{\sqrt{T}}\right).
    \end{equation}
    To bound the second term from Equation \eqref{eq:diff_in_util}, note that Equation \eqref{eq:X_to_Z_ind2} implies that under event $E$,
    \begin{equation}\label{eq:hat_to_X}
        \frob{\hat{X}^t}{\epsilon^t} \le nm \cdot \frob{X^t}{\epsilon^t}.
    \end{equation}
    Finally, combining Equation \eqref{eq:diff_in_util}, Equation \eqref{eq:bound_on_ymu_in_both_cases}, and Equation \eqref{eq:hat_to_X}, we get that conditional on event $E$,
    \[
         \frob{Y^{\mu^*}}{\mu^*} - \frob{X^t}{\mu^*} \le \tilde{O}\left(mn^2\frob{X^t}{\epsilon^t} + \frac{KC_{\mathrm{P}\ref{def:continuous_app}} C_{\mathrm{P}\ref{def:fairness_to_UAR2_app}}n^3m^3}{\sqrt{T}}\right).
    \]
    Conditional on event $E$, the total regret of the algorithm can therefore be bounded by
    \begin{align*}
        \mathrm{Regret}(\alg) &= \sum_{t=0}^{T-1} \left(\frob{Y^{\mu^*}}{\mu^*} - 
 \frob{X^t}{\mu^*}\right) \\
 &\le  \sum_{t=0}^{\log^2(T)\sqrt{T} -1} (b-a) + \sum_{t=\log^2(T)\sqrt{T}}^{T-1} \tilde{O}\left(mn^2\frob{X^t}{\epsilon^t} + \frac{KC_{\mathrm{P}\ref{def:continuous_app}} C_{\mathrm{P}\ref{def:fairness_to_UAR2_app}}n^3m^3}{\sqrt{T}}\right) \\
 &=  \tilde{O}(\sqrt{T})+ \sum_{t=\log^2(T)\sqrt{T}}^{T-1} \tilde{O}\left(mn^2\frob{X^t}{\epsilon^t} + \frac{KC_{\mathrm{P}\ref{def:continuous_app}} C_{\mathrm{P}\ref{def:fairness_to_UAR2_app}}n^3m^3}{\sqrt{T}}\right) \\
 &= \tilde{O}\left( KC_{\mathrm{P}\ref{def:continuous_app}} C_{\mathrm{P}\ref{def:fairness_to_UAR2_app}}n^3m^3\sqrt{T} + mn^2\sum_{t= \log^2(T)\sqrt{T}}^{T-1} \frob{X^t}{\epsilon^t}\right) \\
 &= \tilde{O} \left(KC_{\mathrm{P}\ref{def:continuous_app}} C_{\mathrm{P}\ref{def:fairness_to_UAR2_app}}n^3m^3\sqrt{T} + mn^2\sum_{i,k} \sum_{t= \log^2(T)\sqrt{T}}^{T-1} X_{ik}^t\epsilon_{ik}^t\right)\\
        &= \tilde{O} \left(KC_{\mathrm{P}\ref{def:continuous_app}} C_{\mathrm{P}\ref{def:fairness_to_UAR2_app}}n^3m^3\sqrt{T} + \sum_{i,k} mn^2\sum_{t= \log^2(T)\sqrt{T}}^{T-1} \frac{X_{ik}^t}{\sqrt{N_{ik}^t}}\right). && \text{[Def of $\epsilon_{ik}$]} \numberthis \label{eq:total_regret_under_E}
    \end{align*}

We will focus on bounding the inner summation term $\sum_{t= \log^2(T)\sqrt{T}}^{T-1} \frac{X_{ik}^t}{\sqrt{N_{ik}^t}}$. Define
\[
    E' = \bigcap_{t = \log^2(T)\sqrt{T}}^{T-1} \bigcap_{i,k} \left\{ N_{ik}^t \ge \frac{1}{m} \cdot \sum_{s=\log^2(T)\sqrt{T}}^{t-1} X_{ik}^s \right\}.
\]
\begin{lemma}\label{lemma:prob_of_event_Eprime}
    Using the notation above, $\Pr(E') \ge 1-\frac{1}{3T}$.
\end{lemma}
The proof of Lemma \ref{lemma:prob_of_event_Eprime} can be found in Appendix \ref{app:prob_of_event_Eprime}.

For $t \ge \log^2(T)\sqrt{T}$, define
\[
    f(t) = \left\lfloor \frac{1}{m}\sum_{s=\log^2(T)\sqrt{T}}^{t-1} X_{ik}^s \right\rfloor.
\]
$f(t)$ is an monotonically increasing function with integer outputs and with $f(\log^2(T)\sqrt{T}) = 0$. Consider a fixed $\ell \in [0:f(T-1)]$. Define $s_1$ as the smallest value of $t$ such that $f(t) = \ell$ and $s_2$ as the largest value of $t$ such that $f(t) = \ell$. Then it must be the case that
\[
    \sum_{t=s_1}^{s_2 - 1} X_{ik}^t < m,
\]
as otherwise $f(s_2) \ge \ell+1$. Because $0 \le X_{ik}^t \le 1$, this implies that
\begin{equation}\label{eq:sum_over_Xik_ell}
    \sum_{t : f(t) = \ell } X_{ik}^t = \sum_{t = s_1}^{s_2} X_{ik}^t =  X_{ik}^{s_2} + \sum_{t = s_1}^{s_2-1} X_{ik}^t \le m+1.
\end{equation}
Under event $E'$, for all $i\in[n],k \in [m]$, 
\begin{align*}
        \sum_{t=\log^2(T)\sqrt{T}}^{T-1} \frac{X_{ik}^t}{\sqrt{N_{ik}^t}} &\le \sum_{t=\log^2(T)\sqrt{T}}^{T-1} \frac{X_{ik}^t}{\sqrt{ \lfloor \frac{1}{m}\sum_{s=\log^2(T)\sqrt{T}}^{t-1} X_{ik}^s \rfloor }}  \\
         &= \sum_{t=\log^2(T)\sqrt{T}}^{T-1} \frac{X_{ik}^t}{\sqrt{ f(t)}}  \\
        &= \sum_{\ell=0}^{f(T-1)} \sum_{t : f(t) = \ell } \frac{X_{ik}^t}{\sqrt{f(t)}} \\
        &= \sum_{\ell=0}^{f(T-1)} \sum_{t : f(t) = \ell } \frac{X_{ik}^t}{\sqrt{\ell}} \\
        &= \sum_{\ell=0}^{f(T-1)}  \frac{ \sum_{t : f(t) = \ell } X_{ik}^t}{\sqrt{\ell}} \\
        &\le \sum_{\ell=0}^{f(T-1) }  \frac{m+1}{\sqrt{\ell}}  && \text{[Equation \eqref{eq:sum_over_Xik_ell}]}  \\
        &\le \sum_{\ell=0}^{ T}  \frac{m+1}{\sqrt{\ell}} && \text{[$f(T-1) \le T$]}\\
        &\le O(m\sqrt{T}). \numberthis \label{eq:bound_on_X_to_N}
\end{align*}

    Combining Equation \eqref{eq:total_regret_under_E} and \eqref{eq:bound_on_X_to_N}, the total regret conditional on $E \cap E'$ is bounded by $\tilde{O}(KC_{\mathrm{P}\ref{def:continuous_app}} C_{\mathrm{P}\ref{def:fairness_to_UAR2_app}}n^3m^3\sqrt{T})$. Note that every $X^t$ satisfies the desired fairness constraints conditional on event $E$ due to the construction of LP \eqref{eq:lp_with_slack_app}. Therefore, the algorithm satisfies the constraints for all steps $t$ conditional on event $E \cap E'$. Finally, Lemma \ref{lemma:prob_of_event_E} and Lemma \ref{lemma:prob_of_event_Eprime}  with a union bound gives us that $\Pr(E \cap E') \ge 1-1/T$, as desired.

\end{proof}

\subsection{Proof of Lemma \ref{lemma:prob_of_event_E}}\label{app:prob_of_event_E}    
    For sufficiently large $T$, by Hoeffding's inequality we have that
    {\small
    \begin{align*}
        &\Pr\left(\forall t \ge \log^2(T)\sqrt{T}, \: \norm{\epsilon^t}_1 \le nm\sqrt{2nm\log^2(6nmT)/(\sqrt{T})}\right) \\
        &\ge \Pr\left(\forall t \ge \log^2(T)\sqrt{T},\forall i \in [n], k \in [m], \: \epsilon^t_{ik} \le \sqrt{2nm\log^2(6nmT)/(\sqrt{T})}\right) \\
        &= \Pr\left(\forall t \ge \log^2(T)\sqrt{T},\forall i \in [n], k \in [m], \: \sqrt{\log^2(6nmT)/N^t_{ik}} \le \sqrt{2nm\log^2(6nmT)/(\sqrt{T})}\right) \\
        &= \Pr\left(\forall t \ge \log^2(T)\sqrt{T},\forall i \in [n], k \in [m], \: N_{ik}^t \ge \frac{\sqrt{T}}{2nm}\right) \\
        &= \Pr\left(\forall t \ge \log^2(T)\sqrt{T},\forall i \in [n], k \in [m], \: N_{ik}^t \ge \frac{\sqrt{T}}{nm} - \frac{\sqrt{T}}{2nm}\right) \\
        &\ge \Pr\left(\forall t \ge \log^2(T)\sqrt{T}, \forall i \in [n], k \in [m], \: N_{ik}^t \ge \frac{\sqrt{T}}{nm} - \sqrt{0.5\log(3nmT)}T^{1/4}\right)  \quad \quad \quad \quad \quad \text{[Suff large $T$]}\\
        &\ge \Pr\left(\forall i \in [n], k \in [m], \: N_{ik}^{\sqrt{T}} \ge \frac{\sqrt{T}}{nm} - \sqrt{0.5\log(3nmT)}T^{1/4}\right)\quad \quad \quad \quad \quad \quad \quad  \quad \quad \quad \text{[$N_{ik}^t$ monoton. incr]} \\
        &\ge 1-nme^{-2(\sqrt{0.5\log(3nmT)}T^{1/4})^2/\sqrt{T}}. \quad \quad \quad \quad \quad \quad \quad  \text{[Hoeffding's Ineq]} \\
        &= 1-\frac{1}{3T}. \numberthis \label{eq:large_N_sqrtT}
    \end{align*}
    }
    We will next show that 
    \begin{equation}\label{eq:diff_in_mu_sqrtT}
        \Pr\left(\forall i \in [n],k \in [m],t \in [\sqrt{T}\log^2(T): T-1],  |\hat{\mu}_{ik}^t - \mu^*_{ik}| \le \epsilon_{ik}^t\right) \ge 1- \frac{1}{3T}
    \end{equation}
    For any $i,k$, define $r^0_{ik},...r^{T-1}_{ik}$ as follows. For $t < N_{ik}^T$, define $r^t_{ik}$ as player $i$'s realized value for the item in round $s$, where $s$ is the $(t+1)$th time that an item of type $k$ is allocated to player $i$. For $t \ge N_{ik}^T$, let $r^t_{ik}$ be an i.i.d. draw from the distribution of player $i$'s value for an item of type $k$. By this construction, $r^0_{ik},...r^{T-1}_{ik}$ is a sequence of i.i.d. draws from the distribution of player $i$'s value for item of type $k$. For every $i,k$ and $t$, define $\tilde{\mu}^t_{ik} = \frac{1}{t}\sum_{s=0}^{t-1} r^s_{ik}$. Therefore, for every $t$, $\tilde{\mu}^t_{ik}$ is a sum of exactly $t$ i.i.d. random variables.
    
   We now use Hoeffding's Inequality to bound the following probability.  
    \begin{align*}
        &\Pr\left(\forall i \in [n],k \in [m],t \in [T],  |\tilde{\mu}_{ik}^t - \mu^*_{ik}| \le \sqrt{\log^2(6nmT)/t}\right) \\
        &=\Pr\left(\forall i \in [n],k \in [m],t \in [T],  |t\tilde{\mu}_{ik}^t - t\mu^*_{ik}| \le \sqrt{\log^2(6nmT)t}\right) 
    \end{align*}
    $t\tilde{\mu}_{ik}^t$ is a sum of $t+1$ independent and identically distributed sub-Gaussian random variables. Furthermore, $\E[t\tilde{\mu}_{ik}^t] = t\mu^*_{ik}$. Therefore, by an application of Hoeffding's inequality for sub-Guassian random variables, there exists a constant $c > 0$ (where $c$ depends on the sub-Gaussian value distribution) such that
    \begin{align*}
        &\Pr\left(\forall i \in [n],k \in [m],t \in [T],  |t\tilde{\mu}_{ik}^t - t\mu^*_{ik}| \le \sqrt{\log^2(6nmT)t}\right)  \\
        &\ge 1 - 2nmTe^{-c\left(\sqrt{\log^2(6nmT)t}\right)^2/t} && \text{[Hoeffding's Inequality]} \\
        &\ge 1 - 2nmTe^{-c\log^2(6nmT)} \\
        &=  1-2nmT\left(\frac{1}{6nmT}\right)^{c\log(6nmT)} \\
        &\ge 1-\frac{1}{3T}. && \text{[For sufficiently large $T$]}
    \end{align*}

    Now, note that by construction, we have that $\hat{\mu}_{ik}^t = \tilde{\mu}_{ik}^{N^t_{ik}}$ for any $t \le T$. Therefore, we have that
    \begin{align*}
        &\Pr\left(\forall i \in [n],k \in [m],t \in [\sqrt{T}\log^2(T): T-1],  |\hat{\mu}_{ik}^t - \mu^*_{ik}| \le \epsilon_{ik}^t\right) \\
        &=\Pr\left(\forall i \in [n],k \in [m],t \in [\sqrt{T}\log^2(T): T-1],  |\hat{\mu}_{ik}^t - \mu^*_{ik}| \le \sqrt{\log^2(6nmT)/N_{ik}^t}\right) \\
        &\ge \Pr\left(\forall i \in [n],k \in [m],t \in [T],  |\tilde{\mu}_{ik}^t - \mu^*_{ik}| \le \sqrt{\log^2(6nmT)/t}\right) \\
        &\ge 1- \frac{1}{3T} \numberthis \label{eq:last_third}.
    \end{align*}

This completes the proof of Equation \eqref{eq:diff_in_mu_sqrtT}

Combining Equations \eqref{eq:large_N_sqrtT} and \eqref{eq:diff_in_mu_sqrtT} with a union bound, we have that $\Pr(E) \ge 1-\frac{2}{3T}$.

  \subsection{Proof of Lemma \ref{lemma:Xmu}}\label{app:Xmu}
     \begin{proof}  Consider any $\mu' \in B(\hat{\mu}^t, \epsilon^t)$. Because $\nu \in G^t$ and therefore $\nu \in B(\hat{\mu}^t, \epsilon^t)$, we must have $|\mu'_{ik}- \nu_{ik}| \le 2\epsilon^t_{ik}$ for all $i,k$. Therefore, by the Lipschitz continuity assumption, we have that $|B_\ell(\mu')_{ik} - B_\ell(\nu)_{ik}| \le 2K\epsilon^t_{ik}$. If $X^{\nu} = \mathrm{UAR}$, then $\frob{B_\ell(\mu')}{X^\nu} \ge c_\ell$ by Property \ref{def:notion_of_fairness_app}. If $X^{\nu} \ne \mathrm{UAR}$, then by Equation \eqref{eq:close_approx} we have that for sufficiently large $T$ under event $E$,
    \begin{align*}
        &\frob{B_\ell(\mu')}{X^\nu}\\
        &= \frob{B_\ell(\mu')}{X^\nu} - \frob{B_\ell(\nu)}{X^\nu} +  \frob{B_\ell(\nu)}{X^\nu} \\
        &= \frob{B_\ell(\mu') - B_\ell(\nu)}{X^\nu} + \frob{B_\ell(\nu)}{X^\nu} \\
        &\ge  \frob{B_\ell(\nu)}{X^\nu} - 2K\sum_{ik} \epsilon_{ik}^t X^\nu_{ik} && \text{[Property \ref{def:lipschitz}]}\\
        &\ge  \frob{B_\ell(\nu)}{X^\nu} - 2K\sum_{ik} \epsilon_{ik}^t (Y_{ik}^{\nu} + C_{\mathrm{P}\ref{def:fairness_to_UAR2_app}}4K\frob{Y^{\nu}}{\epsilon^t}) && \text{[Equation \eqref{eq:close_approx}]}\\
        &=  \frob{B_\ell(\nu)}{X^\nu} - 2K\sum_{ik} \epsilon_{ik}^t Y^\nu_{ik}  - 8K^2C_{\mathrm{P}\ref{def:fairness_to_UAR2_app}}\sum_{ik} \epsilon_{ik}^t\frob{Y^{\nu}}{\epsilon^t} \\
        &=  \frob{B_\ell(\nu)}{X^\nu} - 2K\sum_{ik} \epsilon_{ik}^t Y^\nu_{ik}  - 8K^2C_{\mathrm{P}\ref{def:fairness_to_UAR2_app}}\frob{Y^{\nu}}{\epsilon^t}\sum_{ik} \epsilon_{ik}^t \\
        &\ge  \frob{B_\ell(\nu)}{X^\nu} - 2K\frob{Y^{\nu}}{\epsilon^t}  - 2K\frob{Y^{\nu}}{\epsilon^t}   && \text{[$\sum_{i,k} \epsilon_{ik}^t \le \tilde{O}(T^{-1/4}) \le \frac{1}{4KC_{\mathrm{P}\ref{def:fairness_to_UAR2_app}}}$]}\\
        &\ge c_\ell. && \text{[Equation \eqref{eq:first_caseB23case}]} 
    \end{align*}
    In both cases, we showed that $\frob{B_\ell(\mu')}{X^\nu} \ge c_\ell$, therefore $X^{\nu}$ satisfies the constraints of LP \eqref{eq:lp_etc_finallp2_app}.

    Now we will show that $X^{\nu}$ satisfies the constraints of LP \eqref{eq:lp_with_slack_app}. $\hat{X}^t$ satisfies the constraints of LP \eqref{lp:ymu_app} for $\mu =\nu$ because $\hat{X}^t$ satisfies the constraints of LP \eqref{eq:lp_etc_finallp2_app} and $\nu \in B(\hat{\mu}^t, \epsilon^t)$. Therefore, because $Y^\nu$ is the best fractional allocation satisfying the constraints of LP \eqref{lp:ymu_app} for $\mu = \nu$, we must have that $\frob{\hat{X}^t}{\nu} \le \frob{Y^{\nu}}{\nu}$. This in turn implies  (using this in the fourth line) that for sufficiently large $T$,
  \begin{align*}
      \frob{X^\nu}{\mu_U} &\ge \frob{X^\nu}{\nu} \\
      &= \frob{Y^\nu}{\nu} +  \frob{X^\nu}{\nu}  - \frob{Y^\nu}{\nu} \\
      &\ge \frob{Y^\nu}{\nu} - 4KC_{\mathrm{P}\ref{def:fairness_to_UAR2_app}} \cdot \frob{Y^\nu}{\epsilon^t} && \text{[Equation \eqref{eq:swclose}]}\\
      &\ge \frob{\hat{X}^t}{\nu} - 4KC_{\mathrm{P}\ref{def:fairness_to_UAR2_app}}\cdot \frob{Y^\nu}{\epsilon^t} \\ 
      &\ge \frob{\hat{X}^t}{\mu_U^t - 2\epsilon^t} - 4KC_{\mathrm{P}\ref{def:fairness_to_UAR2_app}}\cdot \frob{Y^\nu}{\epsilon^t} \\ 
      &\ge  \frob{\hat{X}^t}{\mu_U^t} - 2\frob{\hat{X}^t}{\epsilon^t} - 4KC_{\mathrm{P}\ref{def:fairness_to_UAR2_app}} \cdot \frob{Y^\nu}{\epsilon^t}   \\
      &\ge  \frob{\hat{X}^t}{\mu_U^t} -  4KC_{\mathrm{P}\ref{def:fairness_to_UAR2_app}} \cdot \max_{\mu' \in G^t} \frob{Y^{\mu'}}{\epsilon^t} - 2\frob{\hat{X}^t}{\epsilon^t} && \text{[$\nu \in G^t$]}
  \end{align*}
    Therefore, $X^{\nu}$ also satisfies all of the constraints of LP  \eqref{eq:lp_with_slack_app}. 
    \end{proof}

\subsection{Proof of Lemma \ref{lemma:prob_of_event_Eprime}}\label{app:prob_of_event_Eprime}
For sufficiently large $T$, Hoeffding's inequality give that
\begin{equation}\label{eq:hoeffd1}
    \Pr\left(\forall i,k,  \: N_{ik}^{\log^2(T)\sqrt{T}} \ge \log^2(T)\frac{\sqrt{T}}{nm} - \sqrt{\tfrac{1}{2}\log^2(T)\sqrt{T}\log(6nmT)}\right) \ge 1-\frac{1}{6T}.
\end{equation}

Define
For $t \ge \log^2(T)\sqrt{T}$, define
\[
M_{t} =  N^t_{ik} -  N_{ik}^{\log^2(T)\sqrt{T}} - \frac{1}{m}\sum_{s=\log^2(T)\sqrt{T}}^{t-1} X_{ik}^s.
\]
This $M_t$ is a martingale with respect to the history because for all $t$,
\begin{align*}
    \E\left[M_{t} | H_{t} \right] &= 
    M_{t-1} +  \E\left[1_{i_{t-1} = i, k_{t-1} = k} - \frac{X_{ik}^{t-1}}{m}\right] \\
    &= M_{t-1}.
\end{align*}
And furthermore, we have that for all $t$,
\[
    |M_{t} - M_{t-1}| \le 1.
\]
Therefore, by the Azuma-Hoeffding Concentration Inequality, for all $i,k,t$,
\begin{align*}
    &\Pr\left(N^t_{ik} -  N_{ik}^{\log^2(T)\sqrt{T}} \le \frac{1}{m}\sum_{s=\log^2(T)\sqrt{T}}^{t-1} X_{ik}^t - \sqrt{0.5t\log(6nmT^2)}\right) \\
    &= \Pr\left(M_t \le -\sqrt{0.5t\log(6nmT^2})\right) \\
    &\le \exp\left(-2\frac{0.5t\log(6nmT^2)}{(t-\log^2(T)\sqrt{T})}\right) \\
    &\le \frac{1}{6nmT^2}.
\end{align*}
Using a union bound, we have that
\begin{align*}
    &\Pr\left(\forall i, k,t\in [\log^2(T)\sqrt{T}: T], \: N^t_{ik} -  N_{ik}^{\log^2(T)\sqrt{T}} \ge \frac{1}{m}\sum_{s=\log^2(T)\sqrt{T}}^{t-1} X_{ik}^t - \sqrt{2t\log(6nmT^2)}\right)\\
    &\ge 1-\frac{1}{6T}. \numberthis \label{eq:hoeffd2}
\end{align*}
Using a union bound to combine Equations \eqref{eq:hoeffd1} and \eqref{eq:hoeffd2}, with probability $1-\frac{1}{3T}$ for sufficiently large $T$, the following holds for all $t \in [\log^2(T)\sqrt{T}: T]$ and for all $i,k$.
\begin{align*}
    N_{ik}^t &= N_{ik}^{\log^2(T)\sqrt{T}} + \left(N^t_{ik} -  N_{ik}^{\log^2(T)\sqrt{T}}\right) \\
    &\ge \log^2(T)\frac{\sqrt{T}}{nm} - \sqrt{\tfrac{1}{2}\log^2(T)\sqrt{T}\log(6nmT^2)} + \frac{1}{m}\sum_{s=\log^2(T)\sqrt{T}}^{t-1} X_{ik}^t - \sqrt{0.5t\log(6nmT^2)} \\
    &\ge  \frac{1}{m}\sum_{s=\log^2(T^2)\sqrt{T}}^{t-1} X_{ik}^t. \numberthis \label{eq:Nik_bound}
\end{align*}

Equation \eqref{eq:Nik_bound} implies that
\[
    \Pr(E') \ge 1-\frac{1}{3T}.
\]

\section{Proof of Lemma \ref{def:continuous}}\label{app:proof_of_continuity}
    Recall by Observation \ref{observation:satisfies_lipschitz_and_zeros} that proportionality in expectation satisfies Lipschitz continuity with $K = 1$. Define $\epsilon = \norm{\mu- \mu'}_1 $. Let $Z^{\mu}$ be the optimal solution to 
    \begin{align*}
        \max &  \frob{X}{\mu} \\
        \text{s.t. } & X_i \cdot \mu_i - \frac{1}{n} \ge - \epsilon  \quad \forall i \in [n]  \\
        & \sum_{i} X_{ik} = 1 \quad \forall k   \numberthis \label{eq:lp_Zmu}
    \end{align*} 
    The constraints in LP \eqref{eq:lp_Zmu} imply the following bound.
    \begin{equation}\label{eq:bound_on_slack}
         \sum_{i : Z^{\mu}_i \cdot \mu - \frac{1}{n}  < 0} \left(\frac{1}{n}  - Z^{\mu}_i \cdot \mu_i \right) \le n\epsilon.
    \end{equation}
    Now we construct a fractional allocation $W^\mu$ such that $W^\mu$ satisfies the proportionality in expectation constraints for $\mu$ and such that $W^\mu$ does not have significantly less social welfare than $Z^\mu$. We will have two cases. 
    
    \textbf{Case 1: } Informally, this is the case when the amount of total positive slack of $Z^{\mu}$ in LP \ref{lp:ymu} is small relative to the amount of negative slack. If 
    \begin{equation}\label{eq:uar_case}
        \frac{a}{b}  \cdot \sum_{i : Z^{\mu}_i \cdot \mu - \frac{1}{n}  \ge 0} \left(Z^{\mu}_i \cdot \mu_i - \frac{1}{n} \right)  \le \sum_{i : Z^{\mu}_i \cdot \mu - \frac{1}{n}  < 0} \left(\frac{1}{n}  - Z^{\mu}_i \cdot \mu_i \right),
    \end{equation}
    then let $W^{\mu} = \mathrm{UAR}$. Because $\mathrm{UAR}$ is always a solution to proportionality constraints, we know that $W^{\mu}$ must be a solution to LP \eqref{lp:ymu}. Furthermore, for $W^{\mu} = \mathrm{UAR}$ we have that the loss in social welfare between $Z^{\mu}$ and $W^{\mu}$ is
    \begin{align*}
        \frob{Z^\mu}{\mu}- \frob{W^{\mu}}{\mu}  &= \sum_i \left(Z_i^\mu \cdot \mu_i - \frac{1}{n}\right) \\
        &\le  \sum_{i : Z^{\mu}_i \cdot \mu - \frac{1}{n}  \ge 0} \left(Z^{\mu}_i \cdot \mu_i - \frac{1}{n} \right)  \\
        &\le \frac{bn}{a}\epsilon. && \text{ Equations \eqref{eq:bound_on_slack} and \eqref{eq:uar_case}}.
    \end{align*}
    \textbf{Case 2:} When Equation \eqref{eq:uar_case} does not hold, the amount of positive slack is large relative to the amount of negative slack, i.e.
    \begin{equation}\label{eq:not_uar_case}
        \frac{a}{b}  \cdot \sum_{i : Z^{\mu}_i \cdot \mu - \frac{1}{n}  \ge 0} \left(Z^{\mu}_i \cdot \mu_i - \frac{1}{n} \right)  > \sum_{i : Z^{\mu}_i \cdot \mu - \frac{1}{n}  < 0} \left(\frac{1}{n}  - Z^{\mu}_i \cdot \mu_i \right).
    \end{equation}
    Therefore, we can transfer allocation from the players with positive slack to the players with negative slack. Formally, if Equation \eqref{eq:uar_case} does not hold, then construct $W^{\mu}$ as  
    \begin{equation}\label{eq:wfirst}
        W^{\mu}_i =
            \begin{cases}
                Z^{\mu}_i -  \frac{b}{a} \cdot \frac{\sum_{i' : Z^{\mu}_{i'} \cdot \mu_{i'}- \frac{1}{n}  < 0} \left(\frac{1}{n}  - Z^{\mu}_{i'} \cdot \mu_{i'} \right)}{ \sum_{i' : Z^{\mu}_{i'} \cdot \mu_{i'}- \frac{1}{n}  \ge 0} \left(Z^{\mu}_{i'} \cdot \mu_{i'} - \frac{1}{n} \right)} \cdot \left(Z^{\mu}_{i} \cdot \mu_{i} - \frac{1}{n} \right) \cdot \frac{Z_{i}^\mu}{Z_i^\mu \cdot \mu_i}& \text{if }  Z^{\mu}_{i} \cdot \mu - \frac{1}{n}  \ge 0\\
                Z^{\mu}_{i} + \frac{b}{a} \cdot \left( \frac{1}{n}  - Z^{\mu}_{i} \cdot \mu_{i} \right) \cdot \frac{\sum_{i' :  Z^{\mu}_{i'} \cdot \mu_{i'}- \frac{1}{n}  \ge 0} \left(\left(Z^{\mu}_{i'} \cdot \mu_{i'} - \frac{1}{n} \right)  \cdot \frac{Z_{i'}^\mu}{Z_{i'}^\mu \cdot \mu_{i'}}\right)}{ \sum_{i' : Z^{\mu}_{i'} \cdot \mu_{i'}- \frac{1}{n}  \ge 0} \left(Z^{\mu}_{i'} \cdot \mu_{i'} - \frac{1}{n} \right)} & \text{otherwise}
            \end{cases}
    \end{equation}
    First we show that this $W^\mu$ will satisfy the proportionality constraints. 
    
    For $i$ such that $  Z^{\mu}_i \cdot \mu - \frac{1}{n}  \ge 0$, we have that
    \begin{align*}
        &W_i^\mu \cdot \mu_i - \frac{1}{n} \\
        &=     Z^{\mu}_i \cdot \mu_i -  \frac{b\sum_{i' : Z^{\mu}_{i'} \cdot \mu_{i'}- \frac{1}{n}  < 0} \left(\frac{1}{n}  - Z^{\mu}_{i'} \cdot \mu_{i'} \right)}{ a\sum_{i' : Z^{\mu}_{i'} \cdot \mu_{i'}- \frac{1}{n}  \ge 0} \left(Z^{\mu}_{i'} \cdot \mu_{i'} - \frac{1}{n} \right)} \cdot \left(Z^{\mu}_{i} \cdot \mu_{i} - \frac{1}{n} \right) \cdot \frac{Z_{i}^\mu \cdot \mu_{i}}{Z_i^\mu \cdot \mu_i} - \frac{1}{n}  \\
        &= Z^{\mu}_{i} \cdot \mu_{i} -  \frac{b\sum_{i' : Z^{\mu}_{i'} \cdot \mu_{i'}- \frac{1}{n}  < 0} \left(\frac{1}{n}  - Z^{\mu}_{i'} \cdot \mu_{i'} \right)}{ a\sum_{i' : Z^{\mu}_{i'} \cdot \mu_{i'}- \frac{1}{n}  \ge 0} \left(Z^{\mu}_{i'} \cdot \mu_{i'} - \frac{1}{n} \right)} \cdot \left(Z^{\mu}_{i} \cdot \mu_{i} - \frac{1}{n} \right) - \frac{1}{n}  \\
        &> Z^{\mu}_i \cdot \mu_i - \left(Z^{\mu}_i \cdot \mu_i - \frac{1}{n} \right)  - \frac{1}{n}  && \text{[Eq. \eqref{eq:not_uar_case}]} \\
        &= 0.
    \end{align*}

    For any $i$ such that  $  Z^{\mu}_i \cdot \mu - \frac{1}{n}  < 0$,
    \begin{align*}
        &W_i^\mu \cdot \mu_i - \frac{1}{n} \\
        &= Z^{\mu}_i  \cdot \mu_i  + \frac{b}{a} \left(\frac{1}{n}  - Z^{\mu}_i \cdot \mu_i\right) \cdot  \frac{\sum_{i' :  Z^{\mu}_{i'} \cdot \mu_{i'}- \frac{1}{n}  \ge 0} \left(\left(Z^{\mu}_{i'} \cdot \mu_{i'} - \frac{1}{n} \right)  \cdot \frac{Z_{i'}^\mu \cdot \mu_i}{Z_{i'}^\mu \cdot \mu_{i'}}\right)}{ \sum_{i' : Z^{\mu}_{i'} \cdot \mu_{i'}- \frac{1}{n}  \ge 0} \left(Z^{\mu}_{i'} \cdot \mu_{i'} - \frac{1}{n} \right)}   - \frac{1}{n} \\
        &\ge Z^{\mu}_i  \cdot \mu_i  + \left(\frac{1}{n}  - Z^{\mu}_i \cdot \mu_i\right) \cdot  \frac{\sum_{i' :  Z^{\mu}_{i'} \cdot \mu_{i'}- \frac{1}{n}  \ge 0}  \left(Z^{\mu}_{i'} \cdot \mu_{i'} - \frac{1}{n} \right)   }{\sum_{i' : Z^{\mu}_{i'} \cdot \mu_{i'}- \frac{1}{n}  \ge 0} \left(Z^{\mu}_{i'} \cdot \mu_{i'} - \frac{1}{n} \right)}   - \frac{1}{n} \\
        &= Z^{\mu}_i  \cdot \mu_i  +  \left(\frac{1}{n}  - Z^{\mu}_i \cdot \mu_i\right) - \frac{1}{n}  \\
        &= 0.
    \end{align*}
    Therefore, $W^\mu$ satisfies the proportionality constraints and is a solution to LP \eqref{lp:ymu}. The only players in $W^{\mu}$ that lose social welfare relative to $Z^{\mu}$ are the players $i$ that give away allocation, i.e. the players $i$ such that $Z^{\mu}_{i} \cdot \mu - \frac{1}{n}  \ge 0$. Therefore, we have that
    \begin{align*}
        & \frob{W^\mu}{\mu} \\
        &\ge \frob{Z^\mu}{\mu} - \sum_{i :  Z^{\mu}_{i} \cdot \mu - \frac{1}{n}  \ge 0}   \frac{b\sum_{i' : Z^{\mu}_{i'} \cdot \mu_{i'}- \frac{1}{n}  < 0} \left(\frac{1}{n}  - Z^{\mu}_{i'} \cdot \mu_{i'} \right)}{ a\sum_{i' : Z^{\mu}_{i'} \cdot \mu_{i'}- \frac{1}{n}  \ge 0} \left(Z^{\mu}_{i'} \cdot \mu_{i'} - \frac{1}{n} \right)} \cdot \left(Z^{\mu}_i \cdot \mu_i - \frac{1}{n} \right)\cdot \frac{Z_i^\mu \cdot \mu_i}{Z_i^\mu \cdot \mu_i} \\
        &= \frob{Z^\mu}{\mu} - \sum_{i :  Z^{\mu}_{i} \cdot \mu - \frac{1}{n}  \ge 0}   \frac{b\sum_{i' : Z^{\mu}_{i'} \cdot \mu_{i'}- \frac{1}{n}  < 0} \left(\frac{1}{n}  - Z^{\mu}_{i'} \cdot \mu_{i'} \right)}{ a\sum_{i' : Z^{\mu}_{i'} \cdot \mu_{i'}- \frac{1}{n}  \ge 0} \left(Z^{\mu}_{i'} \cdot \mu_{i'} - \frac{1}{n} \right)} \cdot \left(Z^{\mu}_i \cdot \mu_i - \frac{1}{n} \right)  \\
        &= \frob{Z^\mu}{\mu} -\frac{b}{a}\sum_{i' : Z^{\mu}_{i'} \cdot \mu_{i'} - \frac{1}{n}  < 0} \left(\frac{1}{n} - Z^{\mu}_{i'} \cdot \mu_{i'}   \right)  \\
        &\ge \frob{Z^{\mu}}{\mu} - \frac{bn}{a}\epsilon. && \text{[Eq \eqref{eq:bound_on_slack}]}
    \end{align*}

\vspace{5mm}
    
    Therefore, we have shown that $W^\mu$ satisfies the constraints of LP \eqref{lp:ymu} and that $\frob{W^\mu}{\mu}  \ge \frob{Z^{\mu}}{\mu} -  \frac{bn}{a}\epsilon$. Recall that $Y^{\mu}$ is the optimal fractional allocation for LP \eqref{lp:ymu} and therefore $\frob{Y^{\mu}}{\mu} \ge \frob{W^{\mu}}{\mu}$. Together, this implies that
   \begin{equation}\label{eq:changeinrhs}
        \frob{Y^{\mu}}{\mu} \geq \frob{W^\mu}{\mu} \geq \frob{Z^{\mu}}{\mu} - \frac{bn}{a}\epsilon.
    \end{equation}
    Furthermore, 
    \begin{align*}
      Y^{\mu'}_i \cdot \mu_i  - \frac{1}{n} &= Y^{\mu'}_i \cdot \mu_i   -Y^{\mu'}_i \cdot \mu'_i +Y^{\mu'}_i \cdot \mu'_i - \frac{1}{n}\\ 
        &= Y_i^{\mu'} \cdot (\mu_i - \mu_i') + Y^{\mu'}_i \cdot \mu'_i - \frac{1}{n}\\
        &\ge Y_i^{\mu'} \cdot (\mu_i - \mu_i') \\
        &\ge -\norm{\mu_i - \mu_i'}_{1} \\
        &\ge -\epsilon.
    \end{align*}
    Therefore, $Y^{\mu'}$ is a solution to LP \eqref{eq:lp_Zmu} and therefore $\frob{Z^{\mu}}{\mu} \ge \frob{Y^{\mu'}}{\mu}$. Combining this with Equation \eqref{eq:changeinrhs}, we have that
    \begin{align*}
        \frob{Y^{\mu}}{\mu} &\ge \frob{Y^{\mu'}}{\mu}  - \frac{bn}{a}\epsilon \\
        &= \frob{Y^{\mu'}}{\mu'} + \frob{Y^{\mu'}}{\mu - \mu'}  - \frac{bn}{a}\epsilon\\
        &\ge \frob{Y^{\mu'}}{\mu'} - \frac{bn}{a}\epsilon. && \text{[$(Y^{\mu'})_{ik} \le 1$ and $\norm{\mu - \mu'}_1 \le \epsilon$]}
    \end{align*}
    By symmetry, the same argument works in the reverse direction to show that $\frob{Y^{\mu'}}{\mu'} \ge \frob{Y^{\mu}}{\mu} - \frac{bn}{a}\epsilon $. Combining both directions proves the desired result.

\section{Proof of Lemma \ref{def:fairness_to_UAR2}}\label{app:proof_of_fairness_to_UAR2}

To show  Lemma \ref{def:fairness_to_UAR2}, we will use the same construction as used to prove Lemma 2 in \cite{procaccia2024honor}.

For $i \in [n]$, define 
\begin{equation}\label{eq:def_of_S}
    S_i =  Y_i^\mu \cdot \mu_i - \frac{1}{n}.
\end{equation}
 $S_i$ can be viewed as the amount of slack of player $i$'s proportionality constraint for $Y^\mu$. We consider the following two cases.

\textbf{Case 1:} $\sum_{i=1}^n S_i \le \frac{b}{a}n\gamma$

In this case we take $X' = \mathrm{UAR}$, which causes at most $\frac{nb}{a}\gamma$ decrease in social welfare relative to $Y^{\mu}$.

\vspace{5mm}
\textbf{Case 2:} $\sum_{i=1}^n S_i > \frac{b}{a}n\gamma$

Define
\begin{equation}\label{eq:def_of_Y}
    \Delta_{ik} := \frac{Y_{ik}^\mu}{\sum_{k'=1}^m Y_{ik'}^\mu} \cdot  \frac{S_i}{\sum_{i'=1}^n S_{i'}} \cdot \frac{n\gamma}{a}.
\end{equation}
Define $X'$ as
\begin{equation}\label{eq:def_of_Xprime}
    X'_{ik} := Y^\mu_{ik} - \Delta_{ik} + \frac{1}{n}\sum_{i'=1}^n \Delta_{i'k}.
\end{equation}
Then as in \cite{procaccia2024honor}, we have that
\[
    X_i' \cdot \mu_i -  \frac{1}{n} \ge \gamma
\]
and
\begin{equation}\label{eq:sum_of_Ys}
 \sum_{i=1}^n \sum_{k=1}^m \Delta_{ik} = \frac{n\gamma}{a}.
\end{equation}
As in \cite{procaccia2024honor}, Equation \eqref{eq:sum_of_Ys} implies that the change in social welfare is upper bounded by $\frac{bn\gamma}{a}$. Furthermore, we also have that
\begin{align*}
    |X'_{ik} - Y^\mu_{ik}| &\le \max \left( \Delta_{ik}, \frac{1}{n}\sum_{i'=1}^n \Delta_{i'k}\right)  && \text{[Equation \eqref{eq:def_of_Xprime}]}\\
    &\le  \sum_{k'=1}^m \sum_{i'=1}^n \Delta_{i'k'} \\
    &\le \frac{n\gamma}{a}. && \text{[ Equation \eqref{eq:sum_of_Ys}]}
\end{align*}

\section{Proof of Theorem \ref{thm:lower_bounds}}\label{app:lower_bounds}

\begin{proof}
    Note that the proof structure is similar to that of Theorem 2 in \cite{procaccia2024honor}. However, the construction provided below has more players and more items, which results in a more complex proof.

    Let $a = 1$, $b = 40$, $n = 3$, and $m = 3$. Furthermore, assume that all values come from normal distributions with variance $1$. We provide a proof by contradiction. First, assume there exists an algorithm $\alg$ such that for any $\mu^* \in [a,b]^{3 \times 3}$, with probability at least $1-1/T$ $\alg$ satisfies the envy-freeness constraints and has regret of less than $\frac{T^{2/3}}{\log(T)}$. 

 Next, consider the following two mean value matrices, where rows represent players and columns represent items. Let $\epsilon = T^{-1/3}$.
    
    \begin{minipage}[t]{.5\linewidth}
    \vspace{0.5pt}
    \begin{center}
    $\mu_1$ = $\begin{bmatrix}
            \frac{20}{42} & \frac{21}{42} & \frac{1}{42}   \\
          \frac{ 19}{42} & \frac{19}{42} & \frac{4}{42} \\
           \frac{1}{42} & \frac{1}{42} & \frac{40}{42}
            \end{bmatrix}$
    \end{center}
    \end{minipage}
    \begin{minipage}[t]{.5\linewidth}
    \vspace{0.5pt}
    \begin{center}
        $\mu_2 = \begin{bmatrix}
            \frac{20}{42} & \frac{21}{42} & \frac{1}{42}\\
            \frac{19}{42} & \frac{19}{42} + \epsilon & \frac{4}{42} - \epsilon \\
            \frac{1}{42} & \frac{1}{42} & \frac{40}{42}
            \end{bmatrix}$
    \vspace{0.5pt}
    \end{center}
    \end{minipage}
    
    Let $P_1$ be the distribution of $H_T$ for algorithm $\alg$ when $\mu^* = \mu_1$, and $P_2$ likewise for $\mu^* = \mu_2$.

    The following lemma will help bound the KL-divergence between $P_1$ and $P_2$.
    \begin{lemma}\label{lemma:bounded_NT}
        Under the proof by contradiction assumption stated above,
        \[
            \E_{P_1}[N_{23}^T + N_{22}^T] \le T^{2/3}
        \]
        and
        \begin{equation}\label{eq:lower_bound_prob_of_small_X22}
             \Pr_{P_1}\left( \sum_{t=0}^{T-1} (X^t_{22} + X^t_{23} + X^t_{32}) > \frac{42T^{2/3}}{\log(T)}\right) < 1/8.
        \end{equation}
    \end{lemma}
    See Appendix \ref{app:lemma_bounded_NT} for the proof of Lemma \ref{lemma:bounded_NT}. Using Lemma \ref{lemma:bounded_NT}, the KL-divergence between $P_1$ and $P_2$ is
    \begin{align*}
        &\mathrm{KL}(P_1, P_2) \\
        &= \E_{P_1}[N^T_{22}] \cdot \mathrm{KL}\Big(N(\frac{19}{42},1), N(\frac{19}{42}+\epsilon, 1)\Big) +  \E_{P_1}[N^T_{23}] \cdot \mathrm{KL}\Big(N(\frac{4}{42},1), N(\frac{4}{42}-\epsilon, 1)\Big) \\
        &= \frac{\E_{P_1}[N^T_{22}]\epsilon^2}{2}+ \frac{\E_{P_1}[N^T_{23}]\epsilon^2}{2} \\
        &=  \frac{\E_{P_1}[N^T_{23} + N^T_{22}]\epsilon^2}{2} \\
        &\le \frac{1}{2}. \qquad \qquad \qquad \text{[Lemma \ref{lemma:bounded_NT}]} \numberthis \label{eq:KL_div_of_P1_P2}
    \end{align*}
        The following lemma is a result of the Bretagnolle-Huber inequality and is proven in \cite{procaccia2024honor}.
    \begin{lemma}\label{lemma:event_complement_total_variation}
        For any two probability distributions $p$ and $q$ defined on the same space and for any measurable event $F$ in this space, $p(F^C) + q(F) \ge \frac{1}{2}e^{-\mathrm{KL}(p,q)}$.
    \end{lemma}

    If we let $p = P_1, q = P_2$, and $F =  \left\{\sum_{t=0}^{T-1} (X^t_{22} + X^t_{23}+ X^t_{32}) \leq \frac{42T^{2/3}}{\log(T)}\right\}$, we can then observe that
    {\small
    \begin{align*}
         &\Pr_{P_1}\left(\sum_{t=0}^{T-1} (X^t_{22} + X^t_{23}+ X^t_{32}) > \frac{42T^{2/3}}{\log(T)}\right) + \Pr_{P_2}\left(\sum_{t=0}^{T-1} (X^t_{22} + X^t_{23}+ X^t_{32}) \le \frac{42T^{2/3}}{\log(T)}\right) \\
         &\ge \frac{1}{2}e^{-\mathrm{KL}( P_1, P_2)} && \text{[Lemma \ref{lemma:event_complement_total_variation}]} \\
         &\ge 1/4. && \text{[Eq \eqref{eq:KL_div_of_P1_P2}]} \numberthis \label{eq:event_complement_lemma_application}
    \end{align*}
    }
    Equation \eqref{eq:lower_bound_prob_of_small_X22} and Equation \eqref{eq:event_complement_lemma_application} together give that
    \begin{equation}\label{eq:small_Xt_under_P2}
         \Pr_{P_2}\left(F\right) = \Pr_{P_2}\left(\sum_{t=0}^{T-1} (X^t_{22} + X^t_{23}+ X_{32}^t) \le \frac{42T^{2/3}}{\log(T)}\right) \ge 1/8.
    \end{equation}
    Under event $F$, there must exist $t \in [0:T-1]$ such that $X^t_{22} + X^t_{23}+ X_{32}^t \le \frac{42T^{-1/3}}{\log(T)}$. We will show that this implies that $X^t$ does not satisfy the envy-freeness in expectation constraints for $\mu_2$. If $X^t_{22} + X^t_{23} + X_{32}^t\le \frac{42T^{-1/3}}{\log(T)}$, then $X_{12}^t \ge 1 - \frac{42T^{-1/3}}{\log(T)}$. For sufficiently large $T$, this implies that player 2's envy in expectation at time $t$ for player 1 is
 \begin{align*}
     &X^t_{11} \cdot \frac{19}{42}  + X^t_{12}\left(\frac{19}{42}+\epsilon\right) + X^t_{13}\left(\frac{4}{42}-\epsilon\right) - X^t_{21} \cdot \frac{19}{42} - X^t_{22}\left(\frac{19}{42}+\epsilon\right) - X^t_{23}\left(\frac{4}{42}-\epsilon\right) \\
     &\ge X^t_{11} \cdot \frac{19}{42}  + X^t_{12}\left(\frac{19}{42}+\epsilon\right) + X^t_{13}\left(\frac{4}{42}-\epsilon\right) - \frac{19}{42} - \frac{42T^{-1/3}}{\log\left(T\right)}\left(\frac{19}{42}+\epsilon\right) - \frac{42T^{-1/3}}{\log\left(T\right)}\left(\frac{4}{42}-\epsilon\right) \\
     &\ge X^t_{12}\left(\frac{19}{42}+\epsilon\right)  - \frac{19}{42} - \frac{42T^{-1/3}}{\log\left(T\right)}\left(\frac{19}{42}+\epsilon\right) - \frac{42T^{-1/3}}{\log\left(T\right)}\left(\frac{4}{42}-\epsilon\right) \\
     &\ge \left(1 - \frac{42T^{-1/3}}{\log\left(T\right)}\right)\left(\frac{19}{42}+\epsilon\right)  - \frac{19}{42} - \frac{42T^{-1/3}}{\log\left(T\right)}\left(\frac{19}{42}+\epsilon\right) - \frac{42T^{-1/3}}{\log\left(T\right)}\left(\frac{4}{42}-\epsilon\right) \\
     &= \epsilon -  \frac{84T^{-1/3}}{\log\left(T\right)}\left(\frac{19}{42}+\epsilon\right) - \frac{42T^{-1/3}}{\log\left(T\right)}\left(\frac{4}{42}-\epsilon\right) \\
     &= T^{-1/3} - \frac{84T^{-1/3}}{\log\left(T\right)}\left(\frac{19}{42}+T^{-1/3}\right) - \frac{42T^{-1/3}}{\log\left(T\right)}\left(\frac{4}{42}-T^{-1/3}\right) \\
    &> 0.
 \end{align*}
 Therefore, the envy-freeness constraints are not satisfied. This in turn implies with Equation \eqref{eq:small_Xt_under_P2} that
 \begin{equation}
     \Pr_{P_2}(\text{EFE for $\mu_2$ not satisfied}) \ge  \Pr_{P_2}\left(F\right) = \Pr_{P_2}\left(\sum_{t=0}^{T-1} X^t_{22}+ X^t_{23} + X_{32}^t \le \frac{42T^{2/3}}{\log(T)}\right) \ge 1/8.
 \end{equation}
 Finally, this completes the proof by contradiction because $\alg$ does not satisfy the envy-freeness constraints for $\mu^* = \mu_2$ with probability at least $1-1/T$. 

\end{proof}

\subsection{Proof of Lemma \ref{lemma:bounded_NT}}\label{app:lemma_bounded_NT}

  \begin{proof}
     First, we will let $E$ be the event that $\alg$ both satisfies the envy-free in expectation constraints for $\mu_1$ and that the regret of $\alg$ is bounded by $\frac{T^{2/3}}{\log(T)}$ when $\mu^* = \mu_1$. Recall that the assumption in the proof by contradiction implies that $\Pr_{P_1}(E) \ge 1-1/T$. 
     
     When $\mu^* = \mu_1$, the allocation that maximizes expected social welfare subject to the envy-freeness constraints is
    \[
        Y^{\mu_1} = \begin{bmatrix}
        0 & 1 & 0 \\
        1 & 0 & 0 \\
        0 & 0 & 1
        \end{bmatrix}.
    \]  
    In order for $X^t$ to satisfy envy-freeness in expectation constraints for $\mu_1$, the following equation must hold. 
    \begin{equation}
        19X^t_{21} + 19X^t_{22} + 4X_{23} \ge 19X^t_{11} + 19X^t_{12} + 4X_{13}.
    \end{equation}
    Substituting, we get that this is equivalent to
    \begin{equation}
        19X^t_{21} + 19X^t_{22} \ge 19(1 - X^t_{21} - X^t_{31}) + 19(1 - X^t_{22} - X^t_{32})  - 4X_{23} + 4X_{13},
    \end{equation}
    and simplifying gives   
    \begin{equation}\label{eq:total_items_to_player_2}
        X^t_{21} + X^t_{22} \ge 1 - \frac{2}{19}X_{23} + \frac{2}{19}X_{13} - \frac{1}{2}X^t_{31} - \frac{1}{2}X^t_{32}.
    \end{equation}
    Using the above equation, we can show that the regret at time $t$ for any $X^t$ that satisfies the envy-freeness in expectation constraints for $\mu_1$ is 
    {\small
    \begin{align*}
        &\frob{Y^{\mu_1}}{\mu^*} - \frob{X^t}{\mu^*}\\
        &= \frac{1}{42} \left(80 - (20X^t_{11} + 21X^t_{12} + X^t_{13} + 19X^t_{21} + 19X^t_{22} + 4X^t_{23} + X^t_{31} + X^t_{32} + 40X^t_{33}) \right)\\
        &= \frac{1}{42} \left(80 - (20(1 - X_{21}^t - X_{31}^t) + 21(1 - X_{22}^t - X_{32}^t) + X^t_{13} + 19X^t_{21} + 19X^t_{22} + 4X^t_{23} + X^t_{31} + X^t_{32} + 40X^t_{33})  \right)\\
        &= \frac{1}{42} \left(39 - (-(X_{21}^t  - X_{22}^t) - X_{22}^t + X^t_{13} + 4X^t_{23} -  19X^t_{31} -20 X^t_{32} + 40X^t_{33}) \right) \\
        &\ge\frac{1}{42} \left( 39 - \left(-\left(1 - \frac{2}{19}X^t_{23} + \frac{2}{19}X^t_{13} - \frac{1}{2}X^t_{31} - \frac{1}{2}X^t_{32}\right) - X_{22}^t + X^t_{13} + 4X^t_{23} -  19X^t_{31} -20 X^t_{32} + 40X^t_{33}\right) \right)\\
        &=\frac{1}{42} \left( 40 - \left(\frac{78}{19}X^t_{23} + \frac{17}{19}X^t_{13} - X_{22}^t -  18.5X^t_{31} -19.5 X^t_{32} + 40X^t_{33}\right)  \right)\\
        &=\frac{1}{42} \left( 40 - \left(\frac{78}{19}X^t_{23} + \frac{17}{19}X^t_{13} - X_{22}^t -  18.5X^t_{31} -19.5 X^t_{32} + 40(1-X^t_{13} - X^t_{23})\right) \right) \\
        &=\frac{1}{42} \left( \frac{743}{19}X^t_{13} + \frac{682}{19}X^t_{23} + X_{22}^t + 18.5X^t_{31} + 19.5X^t_{32} \right) \\
        &\ge\frac{1}{42} \left( X_{23}^t + X_{22}^t + X_{32}^t \right).
    \end{align*}
    }
    Therefore, the regret of $\alg$ if $\alg$ satisfies envy-freeness in expectation for $\mu_1$ is
    \[
       T \cdot  \frob{Y^{\mu_1}}{\mu^*} - \sum_{t=0}^{T-1} \frob{X^t}{\mu^*} \ge \frac{1}{42}\sum_{t=0}^{T-1} X_{23}^t + X^t_{22}+ X_{32}^t.
    \]
    This implies that under event $E$,
    \begin{equation}\label{eq:N22_under_Emu1}
        \sum_{t=0}^{T-1} X_{23}^t +X^t_{22}+ X_{32}^t \le  \frac{42T^{2/3}}{\log(T)}.
    \end{equation}
    Equation \eqref{eq:N22_under_Emu1} implies that
    \[
    \E_{P_1}[N_{23}^T + N_{22}^T \mid E] =    \E_{P_1}\left[\sum_{t=0}^{T-1} X_{23}^t + X^t_{22}\mid E\right] \le  \E_{P_1}\left[\sum_{t=0}^{T-1} X_{23}^t + X^t_{22} + X^t_{32} \mid E\right] \le \frac{42T^{2/3}}{\log(T)}.
    \]

    Therefore, for large enough $T$,
    \begin{align*}
        \E_{P_1}[N_{23}^T + N_{22}^T] &= \E_{P_1}[N_{23}^T + N_{22}^T \mid E]\Pr_{P_1}(E) +  \E_{P_1}[N_{23}^T + N_{22}^T | \neg E]\Pr_{P_1}( \neg E) \\
        &\le \E_{P_1}[N_{23}^T + N_{22}^T \mid E] + T \cdot \frac{1}{T} \\
        &\le \frac{42T^{2/3}}{\log(T)} + 1 \\
        &\le T^{2/3}. \numberthis \label{eq:expected_N22}
    \end{align*}
    This is the first of the two desired equations. For the second equation, note that Equation \eqref{eq:N22_under_Emu1} also implies that for sufficiently large $T$,
    \begin{equation}
      \Pr_{P_1}\left( \sum_{t=0}^{T-1} (X^t_{23} + X^t_{22}+ X_{32}^t) > \frac{42T^{2/3}}{\log(T)}\right) \le \Pr_{P_1}(\neg E) \le \frac{1}{T} < \frac{1}{8}.
    \end{equation}
    \end{proof}